\documentclass[sigconf,authorversion]{acmart}

\usepackage{graphicx}
\usepackage{bm}
\usepackage{amsmath}
\usepackage{amsfonts}
\usepackage{amsthm}
\usepackage{dsfont}
\usepackage{subcaption}
\usepackage{cleveref}
\usepackage{thmtools, thm-restate}
\usepackage[ruled,vlined,linesnumbered]{algorithm2e}
\usepackage[normalem]{ulem}
\usepackage{enumitem}
\setlist{leftmargin=2mm,nosep,itemindent=0mm,labelsep=0.75mm}
\usepackage{booktabs}

\definecolor{darkred}{RGB}{200,0,0}

\newtheorem{theorem}{Theorem}

\newcommand{\swap}[2]{\ensuremath{#1 \rightarrow #2}} 

\newcommand{\algname}{\textsc{Sirius}}
\newcommand{\polaris}{\textsc{Polaris}}

\newcommand{\pars}[1]{\left( #1 \right)}

\newcommand{\C}{\mathcal{C}}
\let\L\relax
\newcommand{\L}{\mathcal{L}}

\newcommand{\V}{\mathcal{V}}
\newcommand{\G}{\mathcal{G}}

\newcommand{\BO}[1]{\mathcal{O}\left(#1\right)}

\newcommand{\BT}[1]{\Theta\pars{ #1 }}

\newcommand{\X}{\mathcal{X}}

\newcommand{\N}{\mathbb{N}}

\newif\ifextversion
\extversiontrue

\AtBeginDocument{%
  }

\copyrightyear{2026}
\acmYear{2026}
\setcopyright{cc}
\setcctype{by-nc-nd}
\acmConference[KDD '26]{Proceedings of the 32nd ACM SIGKDD Conference on Knowledge Discovery and Data Mining V.2}{August 09--13, 2026}{Jeju Island, Republic of Korea}
\acmBooktitle{Proceedings of the 32nd ACM SIGKDD Conference on Knowledge Discovery and Data Mining V.2 (KDD '26), August 09--13, 2026, Jeju Island, Republic of Korea}
\acmDOI{10.1145/3770855.3817950}
\acmISBN{979-8-4007-2259-2/2026/08}
\begin{document}

\title[Sampling Random Graphs from the Colored Configuration Model]{Sampling Random Graphs from the Colored Configuration Model}

\author{Leonardo Pellegrina}
\affiliation{%
  \institution{Department of Information Engineering, University of Padova}
  \streetaddress{Via Gradenigo 6b}
  \city{Padova}
  \country{Italy}
  \postcode{35129}
}
\email{leonardo.pellegrina@unipd.it}

\renewcommand{\shortauthors}{Leonardo Pellegrina}

\begin{abstract}
A fundamental step in knowledge discovery is statistically assessing data mining results. 
In network analysis, 
such evaluation compares the outcome of a given procedure 
with the outcomes obtained from randomized versions of the observed network. 
Despite its importance,
available graph null models only preserve simple characteristics of the observed graph, such as its degree sequence. 

In this paper we introduce the Colored Configuration Model (CCM), a new null model for vertex-colored multigraphs. 
Our main motivation is the study of online social networks, where the color of a user represents their side in a debate. 
The key novelty of CCM is preserving 
the Colored Degree Matrix (CDM), which encodes, for each vertex, the number of neighbors of any given color. 
Preserving the CDM allows fixing the color assortativity of all nodes, e.g., the propensity of each user to interact with other like-minded users. 
This allows testing whether a given phenomenon is explained by the observed CDM, or whether other characteristics of the network might play a key role.
Available graph null models do not preserve the CDM, 
so they cannot assess its impact on real-world tasks, such as testing the significance of network polarization measures. 
To sample from the CCM, we develop 
\algname-B, a simple baseline adapting the Metropolis-Hastings approach,
and \algname, a refined algorithm tailored to 
preserve the CDM, 
thus achieving provably faster mixing. 
In our experimental evaluation, we test \algname\ on real-world networks, 
comparing it with related network null models. 
We observed that the evaluation of the statistical significance of polarization measures with \algname\ may lead to different insights compared to available null models. 
Thus, \algname\ is an effective tool for the statistically-sound analysis of social networks. 
\end{abstract}

\ccsdesc[500]{Information systems~Data mining}
\ccsdesc[500]{Theory of computation~Random network models}
\keywords{Network Analysis, Hypothesis Testing, Null Model, Polarization}


\maketitle


\section{Introduction}
\label{sec:intro}

The process of knowledge discovery involves several phases designed to uncover meaningful information from complex datasets.
A critical aspect of the knowledge discovery pipeline is to avoid reporting uninteresting or spurious discoveries~\cite{tan2016introduction,PellegrinaRV19b,hamalainen2019tutorial,riondato2023statistically}, which may be due to noise or random fluctuations of the data. 
A common approach to prevent this issue is evaluating the \emph{statistical significance} of data mining results, 
leveraging the 
statistical hypothesis testing framework~\cite{fisher1935design}. 
To do so, the observations gathered from a data mining analysis are evaluated 
against their distribution under an appropriate \emph{null hypothesis}, 
i.e., a set of assumptions and a \emph{realistic} model for the underlying \emph{data generating process}: findings that strongly deviate from what expected under the null distribution are flagged as \emph{significant}; results that are likely under the null hypothesis are instead discarded.

In graph analysis, 
the result of a given investigation may be quantified by a measure;
to evaluate its significance, the same measure is repeated on a set of random networks drawn from a null model. 
The most commonly used graph null model is the Configuration Model (CM)~\cite{fosdick2018configuring}, in which edges are randomized, but the degree sequence of the nodes is kept fixed. 
If the measure obtained from the observed network is significantly different w.r.t. the randomized networks, it can confidently be concluded that it is not \emph{explained} by the properties preserved in the model (i.e., the degree sequence). 
Therefore, a key requirement of graph null models is their ability to preserve realistic properties of real-world networks. 

A topical example of network analysis, 
and the main motivation for our work, 
is the study of \emph{structural polarization} in online social networks \cite{conover2011political,cinelli2021echo,gonzalez2023asymmetric}, 
i.e., evaluating if the topology of social interactions between profiles 
has a connectivity concentrated within homogeneous subgroups, creating structural separations between users with different opinions.
This phenomenon may lead to the spread of misinformation and opinion radicalization, and is considered as one of the most pressing risks in our societies~\cite{globalriskrep}. 
The observational studies on structural polarization gather data from online social networks, 
such as collections of tweets covering a given topic; 
this rich source of information 
can be represented with vertex-colored multigraphs~\cite{garimella2018quantifying}. 
Typically, 
nodes are profiles, where
the color of a node denotes their opinion, or their side in a given debate, which may be inferred from the shared or endorsed content (e.g., hashtags included in posts); 
connections (edges) are formed when a pair of profiles interact (e.g., after an endorsement, such as a retweet, or if they follow each other). 
Several measures have been proposed to quantify structural polarization~\cite{guerra2013measure,matakos2017measuring,garimella2018quantifying}, e.g., in terms of the probability that a random walk on the vertex-colored multigraph starts and ends only traversing nodes of the same color. 
To verify whether the observed level of structural polarization is truly meaningful, it is necessary to compare the same measure on random networks from a suitable graph null model~\cite{salloum2022separating}. 

Despite the fundamental importance of this statistical assessment step, the development of realistic null models for \emph{vertex-colored} graphs received scant attention~\cite{salloum2022separating,preti2025polaris}. 
Commonly used graph null models are tailored to unlabelled graphs, and typically preserve basic properties of the observed network, 
thus are unsuited to modeling richer type of networked data.

We contribute to this important line of research with a new null model for colored multigraphs, i.e., graphs where each vertex has a color and edges may appear multiple times. 
Our model, called Colored Configuration Model (CCM), generates random graphs from an ensemble that fixes the colors of each node, and also preserves the Colored Degree Matrix (CDM), 
i.e., the number of neighbors of any color for all nodes. 
In our model, we assume that the color  is an intrinsic characteristic of each node; for instance, the opinion w.r.t. a controversial topic in a social network.
In addition, we also treat the \emph{number of colored neighbors} as an inherent trait of each node: 
this is a proxy of the node's \emph{color assortativity}~\cite{newman2003mixing}, which quantifies its tendency to mainly interact with other nodes with the same color, or its interest in content that may not align with his views, a key property in the study of homophily~\cite{mcpherson2001birds,jasny2015empirical,baumann2020modeling}.

Figure~\ref{fig:cdmexample} shows an example of two colored multigraphs with the same set of colored vertices and the same CDM. 
From the example, it is immediate to observe that both graphs share the same number of colored neighbors, for each node; for instance, the node $2$ has four red and one blue neighbors (taking into account the edge multiplicities); the node $5$ has three red neighbors. 

\begin{figure}[ht]
\begin{subfigure}{.495\textwidth}
  \centering
  \includegraphics[width=\textwidth]{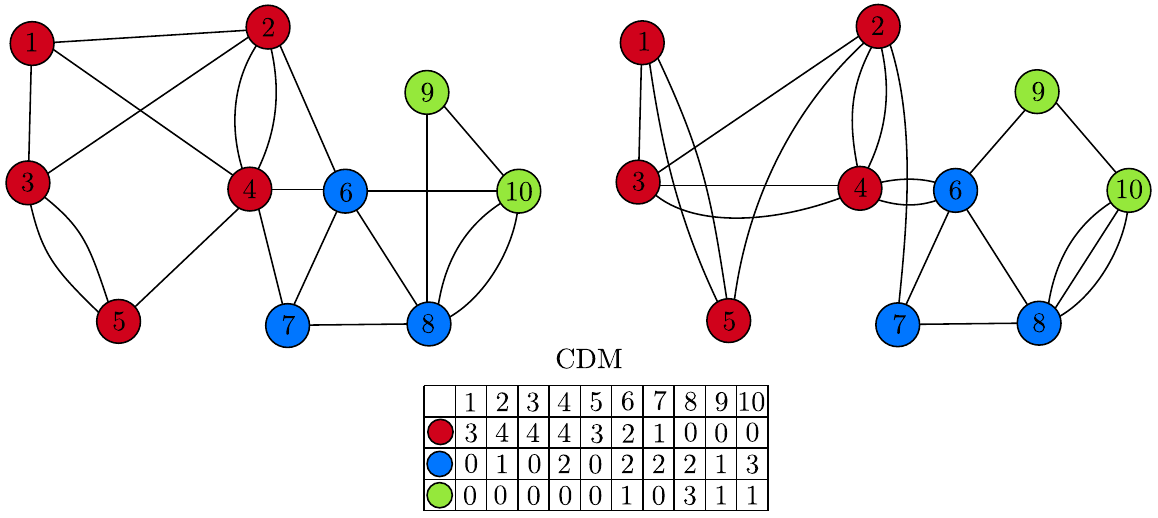} 
\end{subfigure}
\caption{  
Two colored multigraphs with the same Colored Degree Matrix (CDM).
}
\Description{This image shows an example of two colored multigraphs with the same Colored Degree Matrix (CDM).}
\label{fig:cdmexample}
\end{figure}

In a social network setting, the degree of a node quantifies its popularity, while the CDM encodes the 
tendency of each individual to interact with like-minded users, 
or their willingness to interact with users with other points of view. 
For instance, two profiles with the same innate opinion, i.e., the same color, may differ substantially in their propensity to interact with profiles with alternative views (e.g., nodes $3$ and $4$ in the example of Figure~\ref{fig:cdmexample}). 

The null model Polaris, recently presented by~\cite{preti2025polaris}, is the first extension of the CM to colored multigraphs. 
This model  
specifies an ensemble of colored multigraphs with the same Joint Color Matrix (JCM), 
that encodes the \emph{total} number of edges between any pair of colors. 
This model fixes the color assortativity \emph{globally}; 
instead, the CDM preserves this quantity \emph{locally}, i.e., simultaneously for all nodes. 
Fixing the color assortativity only at a global level will likely discard the qualitative characteristics of the local interactions, as they may produce random networks with significantly different proportions of colored neighbors.

Figure~\ref{fig:cdmjcmexample} 
shows an example of two colored multigraphs with the same degree sequence and JCM, but with different CDMs. 
From this example, we may observe that some nodes have considerably dissimilar entries in the two CDMs, i.e., have a sensible discrepancy in the number of colored neighbors. 
For instance, the node $4$ has blue and red neighbors in the left graph, while it has only red neighbors in the right graph;
or, on the contrary, the node $5$ has only red neighbors in the original graphs, while it has both red and blue neighbors in the graph on the right.

\begin{figure}[ht]
\begin{subfigure}{.495\textwidth}
  \centering
  \includegraphics[width=\textwidth]{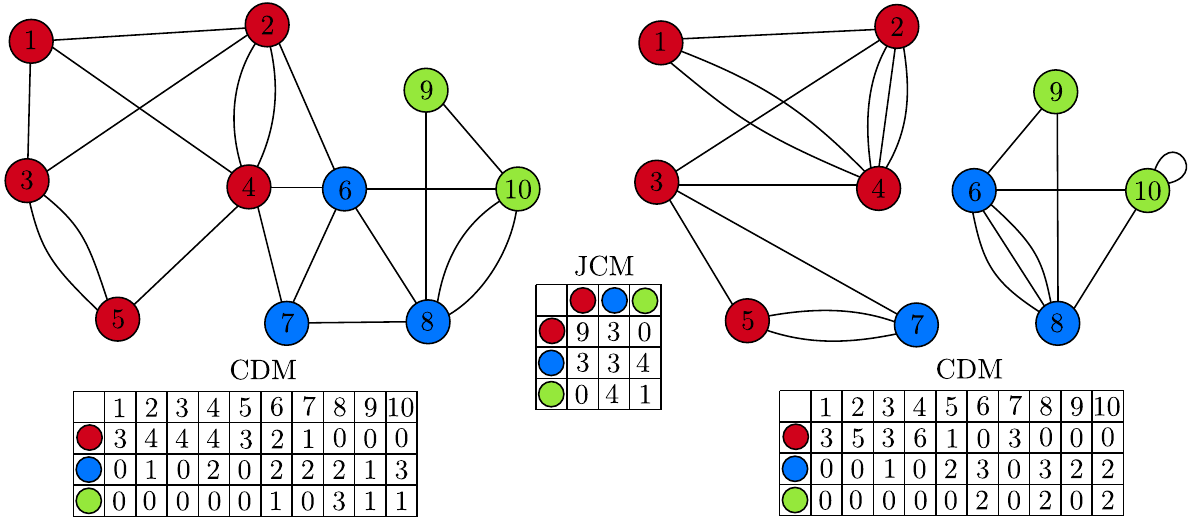} 
\end{subfigure}
\caption{  
Two colored multigraphs with the same degree sequence and Joint Color Matrix (JCM), but different CDMs.
}
\Description{This image shows an example of two colored multigraphs with the same degree sequence and Joint Color Matrix (JCM), but different CDMs. }
\label{fig:cdmjcmexample}
\end{figure}

Preserving the JCM only may generate random networks with local exposure and/or engagement to opinions that are inconsistent w.r.t. the observed graph, as this example shows. 
For instance, a given profile may be inclined to support both views on a given controversial topic, thus he may interact equally with users from two different communities; 
on randomized versions of the graph with fixed JCM the same user may be either over- or under-exposed to any of both communities. 
Depending on the goals of the statistical analysis, and the related assumptions, this behavior may be undesirable. 
Instead, random graphs from CCM preserve the type of interactions of all vertices, i.e., the \emph{local color assortativity}, allowing to test the effect of other network characteristics. 

Our main goal is to formally define a new network null model with constrained CDM, and to design algorithms to efficiently sample random graphs from the model. 
Our algorithms are the first to allow assessing the statistical significance of graph mining tasks  
while controlling for the effect of the color assortativity of each node. We now summarize our contributions:
\begin{itemize} 
\item We introduce the Colored Configuration Model (CCM), 
a new null model for vertex-colored multigraphs. 
The CCM is defined over an ensemble of multigraphs with the same Colored Degree Matrix (CDM), which specifies the number of neighbors of any color for all nodes. 
Our null model is inspired by the study of social networks, as it constraints its graph ensemble to respect the local color assortativity, i.e., the intrinsic tendency of each user in creating ties with similar entities of the network. 

\item To efficiently sample from the CCM, 
we develop the \algname-B and \algname\ algorithms. 
At a high level, our algorithms leverage the Markov Chain Monte Carlo (MCMC) framework, using the Metropolis-Hastings (MH) method~\cite{mitzenmacher2017probability}. 
While \algname-B is a baseline method based on a variant of the standard MH approach to generate networks with fixed degree sequence, \algname\ is a refined approach tailored to the task. 
While both methods converge to the desired target distribution over the space of multigraphs with fixed CDM, \algname\ is provably more effective than its baseline, 
as we show a precise relationship between the mixing time of the two methods. 
We quantify the improvement factor, that depends on the proportions of the colored edges within the network. 

\item We perform an extensive experimental evaluation, considering real-world networks from several domains.
In our experiments, we verify that state-of-the-art graph null models do not preserve the CDM, and that they output networks with considerable discrepancies in the number of colored neighbors. 
We compare the performance of \algname-B with \algname. 
The results show that the latter is significantly more effective, which aligns with our theoretical results. 
Finally, we applied \algname\ to the task of assessing the statistical significance of structural polarization measures, comparing the results with state-of-the-art null models. 
Interestingly, our results indicate that preserving the CDM may or may not explain the level of polarization observed on real-world instances.  
\end{itemize}

\section{Related Works}
\label{sec:relworks}

There is a vast literature on random graph models. 
Arguably, one of the most known and studied is 
the Erd\H{o}s-R\'enyi~\cite{erdos1959random} model,
where any pair of nodes is connected by an edge with some fixed probability $p$. 
The most commonly used graph null model is the 
Configuration Model (CM)~\cite{bollobas1980probabilistic,bender1978asymptotic,fosdick2018configuring}, 
which generates random graphs with a fixed degree sequence. 
In contrast with analytically-tractable models, such as Erd\H{o}s-R\'enyi, 
CM is often the model of choice for practitioners, as several key properties of real-world graphs are constrained by the distribution of vertex degrees \cite{callaway2000network,newman2001random,larremore2011predicting}. 
Several other models have been proposed and analyzed, 
with different characteristics and goals: some examples are the 
preferential attachment~\cite{barabasi1999emergence}, 
exponential~\cite{snijders2002markov}, 
and maximum-entropy~\cite{squartini2017maximum} 
models, 
that either preserve the degree sequence only in expectation, or 
are designed to increase the probability of observing certain subgraphs. 
The use of random graph models for statistical testing have found applications 
in several scientific domains, such as 
sociology~\cite{moreno1938statistics}, 
ecology~\cite{connor1979assembly,harvey1983null}, 
system biology~\cite{milo2002network,itzkovitz2003subgraphs}, 
and economy~\cite{sundararajan2008local}.

Our algorithms address the problem of generating random graphs from an extension of CM tailored to vertex-colored multigraphs. 
Since no efficient approach for direct sampling from the space of such graphs is known, 
our methods leverage the Markov Chain Monte Carlo (MCMC) framework~\cite{mitzenmacher2017probability} and the Metropolis-Hastings (MH)~\cite{metropolis1953equation,hastings1970monte} approach. 
The MCMC is a general technique to explore the sample space with a random walk, where 
the transition probabilities are set to converge to a target distribution, such as the uniform one.  
For graph sampling, each step of the random walk modifies the current graph, while preserving the characteristics of interest (e.g., the degree sequence). 
If the random walk is long enough, i.e., a sufficiently large number of steps is performed, and if the Markov chain satisfies some properties (it is irreducible and aperiodic), 
then the random sequence of explored states converges to the desired target distribution. 
A theoretical challenge is it prove analytical bounds on the mixing time~\cite{mitzenmacher2017probability}, i.e., the number of steps required to (approximately) converge to the stationary distribution (e.g., in terms of the total variation distance). 
While some results are known~\cite{greenhill2014switch,dutta2025sampling}, they are mostly of theoretical interest, as the upper bounds to the number of iterations involve high-degree polynomials in the graph size. 
In practice, several diagnostic methods~\cite{mitzenmacher2017probability,roy2020convergence} are used to assess the convergence of the chain, such as the degree assortativity of the sequence of explored graphs ~\cite{dutta2025sampling}. 

While our goal is to propose new methods to efficiently assess the statistical significance of graph analysis tasks under realistic assumptions, e.g., structural polarization measures, 
other recent works \cite{zhu2021minimizing,haddadan2022reducing,fabbri2022rewiring,zhu2022nearly,adriaens2023minimizing,cinus2023rebalancing,coupette2023reducing} have tackled this issue from different perspectives;
for instance, by proposing modifications to the input graph (e.g., link recommendation) to reduce structural polarization.

\section{Preliminaries}
\label{sec:prelims} 
This section introduces the key concepts useful for our work. 
We denote multisets with double curly braces, e.g., $A = \{\!\!\{ a , a , b , c , c \}\!\!\}$,  
and sets with single curly braces, e.g., 
$B = \{ a , b , c \}$. 
The cardinality of a set (or a multiset) $A$ is denoted with $|A|$; in the examples above, $|A| = 5$ and $|B| = 3$. 

A \emph{colored multigraph} $G$ is defined as $G = (V , E , \L , c)$,
where $V$ is a set of $n$ nodes,
$E$ is a multiset of $m$ edges,
where each edge $(u,v)$ is an unordered pair of vertices, 
$\L$ is a set of colors,
and $c : V \rightarrow \L$ is a labelling function that assign a color to each node of $G$. 
An edge $(u,v)$ is \emph{monochromatic} when $c(u) = c(v)$, and \emph{bichromatic} when $c(u) \neq c(v)$. 
Since $E$ is a multiset, there may be multiple occurrences of the same edge $(u,v)$, which may represent multiple interactions between $u$ and $v$, or the strength of the link between $u$ and $v$. 
The multiset $E$ may also contain self-loops, i.e., edges $(u,u)$ connecting $u$ to itself. 
For any $\ell , r \in \L$, let $E_{\ell , r} = \{\!\!\{ (u , v) \in E : \{ c(u) , c(v) \} = \{ \ell , r \} \}\!\!\}$ be the multiset of edges whose endpoints have colors $\ell , r \in \L$. 
We define $w_E(u,v)$ as the number of occurences of the edge $(u,v)$ in the multiset $E$. 
Let $N_G(v)$ be the multiset of neighbors of $v$ in $G$, 
such that  $u \in N_G(v)$ for each occurence of $(u,v) \in E$;
and let $N_G^\ell(v)$ be the multiset of neighbors of $v$ with color $\ell$ in $G$, 
such that  $u \in N_G^\ell(v)$ for each occurence of $(u,v) \in E$ with $c(u) = \ell$.\footnote{Note that, following the standard convention, for each self-loop $(v,v) \in E$  
there are two occurrences of $v$ in the multisets $N(v)$ and $N_G^{c(v)}(v)$.}  
Then, the degree $d_G(v)$ of $v$ is $d_G(v) = | N_G(v) |$, and the \emph{$\ell$-colored degree} of $v$ is $d_G^\ell(v) = | N_G^\ell(v) |$. 
It holds $d_G(v) = \sum_{\ell \in \L}  d_G^\ell(v)$. 

We now define the Colored Degree Matrix (CDM), a key concept for the graph null model we introduce in this work. 
\begin{definition}
The Colored Degree Matrix (CDM) $\C_G$ of a colored multigraph $G = (V , E , \L , c)$ is a matrix $\C_G \in \N^{|\L| \times n}$ 
where every entry $\C_G[ \ell , v ]$ is equal to the $\ell$-colored degree $d_G^\ell ( v )$ of $v$ in $G$, i.e.,
\begin{align*}
\C_G[ \ell , v ]  = d_G^\ell ( v ), \forall \ell \in \L , v \in V .
\end{align*}
\end{definition}
Let $\X$ be the domain of all colored multigraphs with the same set of colored vertices $V$ of $G$. 
We define the space $\V$ as the set of all multigraphs $H$ with CDM equal to $\C_G$, such that $\V = \{ H \in \X : \C_H = \C_G \}$. 
Figure~\ref{fig:cdmexample} shows an example of two colored multigraphs 
from the same space $\V$, i.e., having the same set of colored vertices and the same CDM.

It is immediate to observe that for two graphs $G,H$ from the space $\V$ 
it holds $[ d_G(1) , d_G(2) , \dots , d_G(n) ] = [ d_H(1) , d_H(2) , \dots , d_H(n) ]$, i.e., they have the same degree sequence, since it corresponds to the column-totals of the CDM. 
However, two graphs with the same degree sequence do not necessarily have the same CDM. 
Similarly, two graphs $G,H \in \V$ have the same Joint Color Matrix (JCM)~\cite{preti2025polaris}, which quantifies the (global) color assortativity of the graph. 
The JCM $J_G \in \N^{|\L| \times |\L|}$ is defined as follows: for any $\ell , r \in \L$, each entry $J_G[ \ell , r ]$ is equal to the number of edges $(u,v)$ with $c(u) = \ell$ and $c(v) = r$. 
The matrices $J_G$ and $\C_G$ for the colored multigraph $G$ are related:
the matrix $J_G$ is obtained from the CDM $\C_G$ by summing its entries of as follows: 
\begin{align*}
J_G[\ell , \ell] = \frac{1}{2} \!\!\! \sum_{v : c(v) = \ell} \!\!\!\! \C_G[ \ell , v ] , \hspace{0.4cm} 
J_G[\ell , r] = \!\!\!\! \sum_{v : c(v) = r} \!\!\!\! \C_G[ \ell , v ] , 
\hspace{0.1cm}
\ell \neq r .
\end{align*}
While sharing the same CDM implies having the same JCM, two graphs with the same  degree sequence and the same JCM may not have the same CDM, as shown in Figure~\ref{fig:cdmjcmexample}.

Given a probability distribution $\pi$ over the space $\V$,
our goal is to design an efficient approach to sample a multigraph from $\V$ according to $\pi$. 
Our algorithms are based on the Markov Chain Monte Carlo (MCMC) scheme, and leverage the Metropolis-Hastings (MH) method to traverse the Markov chain and to guarantee that the stationary distribution is $\pi$. 
We now introduce these concepts, while in Section~\ref{sec:algo} we show how to effectively adapt them to our problem.

\textbf{MCMC and MH methods.} 
A \emph{Markov chain} can be represented by a directed, weighted graph
$\G = ( \mathcal{S} , \mathcal{E} , \omega )$, 
where $\mathcal{S}$ is a set of \emph{states} of the chain. 
For a pair of states $i,j \in \mathcal{S}$, there exist a directed edge $(i,j) \in \mathcal{E}$ if it is possible to move from $i$ to $j$ in one step.
In such a case, $j$ is a \emph{neighbor} of $i$, and $\omega(i,j)$ is the \emph{weight} of the edge $(i,j)$, with $\omega(i,j) > 0$. 
Self-loops $(i,i)$, i.e., moves from the state $i$ to the same state $i$, are allowed. 
The weights $\omega(i,j)$ of outgoing edges from a state $i$ represent the one-step transition probabilities of moving from $i$ to $j$ in one step, and satisfy $\sum_j \omega(i,j) = 1$. 
Equivalently, the Markov chain $\G$ can be represented by a \emph{transition matrix} $P$, where each entry $P_{i,j}=\omega(i,j)$ is the probability that the chain moves from state $i$ to state $j$ in one step. 

A key application of Markov chains is the task of sampling states of $\mathcal{S}$ according to a probability distribution $\pi$;   
this is achieved by performing a suitable \emph{random walk} over $\G$, and reporting in output the state of the chain after a sufficient number of steps, i.e., the \emph{mixing time} of the chain. 
Under some conditions, the sequence of states explored by the random walk converges to an \emph{unique stationary distribution}. 
We now define such conditions. 
A Markov chain $\G$ is \emph{irreducible} if and only if, for every arbitrary pair of states $i,j$, there exist a path from $i$ to $j$ in $\G$. In other words, if and only if $\G$ is strongly connected. 
$\G$ is \emph{aperiodic} if the greatest common divisor of the lengths of its cycles is $1$. 
Any finite, irreducible, and aperiodic Markov chain is \emph{ergodic}, and it has a unique stationary distribution $\bar{\pi}$. 
If the chain if ergodic, a random walk on $\G$ converges to a unique stationary distribution $\bar{\pi}$; 
however, setting the one-step transition probabilities $\omega$ such that $\bar{\pi} = \pi$ is not trivial. 

The MH approach is a generic procedure to sample a state from $\mathcal{S}$, according to $\pi$, using an arbitrary \emph{neighbor proposal distribution} $\xi$. 
For a state $i$ and a neighbor $j$ of $i$, $\xi_i(j)$ is the probability of proposing the step from $i$ to $j$, with $\sum_j \xi_i(j) = 1$. 
The MH approach proceeds as follows: 
given an arbitrary state $i \in \mathcal{S}$, a neighbor state $j$ is sampled with probability $\xi_i(j)$;
the proposed state $j$ is \emph{accepted}, i.e., the chain moves from $i$ to $j$, with probability $\min \left\{ 1 , \frac{\pi(j) \xi_j(i) }{\pi(i) \xi_i(j) } \right\} $, while it stays in state $i$ otherwise. 
Performing this random walk over an ergodic chain $\G$ guarantees that the sequence of explored states converges to the unique stationary distribution $\pi$. 

In the following sections we show how to adapt this approach to the task of sampling random graphs from the CCM. 
Given the desired space $\mathcal{S}$ and the target distribution $\pi$, 
the main steps are: define the Markov chain $\G$; proving that it is strongly connected and aperiodic; design an effective proposal distribution $\xi$.

\section{\algname: Sampling from the CCM}
\label{sec:algo}
In this section we present the CCM model and our methods to generate random networks from the CCM. 
We first define the main operations used by our methods to explore a Markov chain, whose states are multigraphs with prescribed CDM. 
Due to space constraints, all proofs are in the Appendix. 

\textbf{The CCM model.}
We define the Double Edge Swap (DES), 
a key operation to modify a graph preserving its degree sequence. 
\begin{definition}
Given a colored multigraph $G = (V , E , \L , c)$
and a pair of edges $e_1 , e_2 \in E$,
a \emph{Double Edge Swap (DES)} 
is defined as an operation to $G$
that swaps one endpoint of $e_1$ with one endpoint of $e_2$, 
modifying $e_1 , e_2$ to a new pair of edges $e_1^\prime , e_2^\prime$. 
This operation is denoted by $\swap{e_1 , e_2}{ e_1^\prime , e_2^\prime }$. 
\end{definition} 
Given a pair of edges $(u,v) , (x,y)$, there are exactly two unique DES involving them:
\swap{(u,v) , (x,y)}{(u,x) , (v,y)} and 
\swap{(u,v) , (x,y)}{(u,y) , (v,x)}. 
The multigraph $H = (V , E \setminus \{\!\!\{ e_1 , e_2 \}\!\!\} \cup \{\!\!\{ e_1^\prime , e_2^\prime \}\!\!\} , \L , c)$ is obtained by \emph{applying} the DES $\swap{e_1 , e_2}{ e_1^\prime , e_2^\prime }$ to the multigraph $G$.
If $H \neq G$ the DES is \emph{changing}, and is non-changing otherwise. 
For the multigraph on the left of Figure~\ref{fig:cdmexample}, 
for the pair of edges $(2,3)$, $(4,5)$ the two possible DESs are 
\swap{(2,3) , (4,5)}{(2,4) , (3,5)}
and 
\swap{(2,3) , (4,5)}{(2,5) , (3,4)}, which are both changing. 

The DES operation is the key step to modify a multigraph while preserving its degree sequence. 
In MCMC algorithms, applying sequences of DES is equivalent to exploring the space of multigraphs with the same degree sequence~\cite{fosdick2018configuring}. 
However, a DES may not preserve the CDM, so it is not always appropriate to explore the space of multigraphs with the same CDM. 
Therefore, we define the CDES, a DES that always preserves the CDM.

\begin{definition}
Given a colored multigraph $G = (V , E , \L , c)$
and a pair of edges $e_1 , e_2 \in E$,
a \emph{CDM-preserving Double Edge Swap (CDES)} 
is defined as 
a DES $\swap{e_1 , e_2}{ e_1^\prime , e_2^\prime }$ 
that preserves the CDM of $G$.
\end{definition}
Similarly to a DES, a CDES is \emph{changing} if the resulting multigraph $H$ is $H \neq G$, and is non-changing otherwise. 
For the multigraph on the left of Figure~\ref{fig:cdmexample}, 
the DES  
\swap{(2,6) , (4,7)}{(2,7) , (4,6)}
is a changing CDES, but  
\swap{(2,6) , (4,7)}{(2,4) , (6,7)} is a changing DES but not a CDES. 

We now describe a Markov chain for the space of multigraphs with fixed CDM. 
We define  
$\G = ( \mathcal{V} , \mathcal{E} , \omega )$
as a directed, weighted graph. 
Two states $H,Y \in \mathcal{V}$ are connected by a directed edge $(H,Y) \in \mathcal{E}$, with $\omega(H,Y) > 0$, if and only if there exist a CDES that, applied to $H$, yields $Y$. 
For a changing CDES it holds $H \neq Y$, while $H = Y$ for a non-changing CDES.
Therefore, the chain $\G$ may also contain self-loops $(H,H) \in \mathcal{E}$ if there is a non-changing CDES that can be applied to $H$.
Recall that the weight $\omega(H,Y)$ represents the probability that a random walk on the node $H$ moves to the state $Y$, such that $ \sum_{Y \in \V} \omega(H,Y) = 1$. 

We now prove that the Markov chain $\G$ defined above is irreducible and aperiodic, which implies that it is ergodic. 

\begin{theorem}
\label{thm:irreducible}
The graph $\G$ is irreducible.
\end{theorem}

\begin{theorem}
\label{thm:aperiodic}
Given the multigraph $G$, if at least one of the following conditions holds,
then the graph $\G$ is aperiodic:
1) there exist a color $\ell$ such that there exist two monochromatic edges with color $\ell$; or
2) there exist a color $\ell$ and a node $v$ with $c(v) \neq \ell$
with $d^\ell_G(v) \geq 2$.
\end{theorem}

We remark that the mild conditions stated in Theorem~\ref{thm:aperiodic} are sufficient to ensure aperiodicity of $\G$, but not necessary, since the chain may include additional self-loops (from rejections of the MH approach happening with probability $>0$); moreover, only a restricted class of graphs do not satisfy them. 
In any case, our algorithms can be adapted easily to handle instances that do not satisfy the aforementioned conditions, with the following simple modification: 
it is enough to add a self-loop to each vertex $H$ of the chain $\G$, that is chosen at random at each MCMC step with fixed probability, i.e., by considering a \emph{lazy} version of $\G$. 

In the following sections we describe our algorithms \algname\ and \algname-B to sample random graphs from the CCM model. 

\textbf{A baseline algorithm.}
We first introduce \algname-B, which is based on an extension of the standard MH approach to draw random graphs for the Configuration Model (CM)~\cite{fosdick2018configuring}. 
\algname-B is described in Algorithm~\ref{alg:baseline} (in the Appendix). 
The algorithm, for each of the $t$ iterations,
samples uniformly at random a pair $(u,v) , (x,y)$ of distinct edges of $E$ (lines~\ref{algb:sampleone} and~\ref{algb:sampletwo});
then, it samples uniformly at random one of the two possible DES involving these edges.
This is done by generating a uniform random value $p$ in the interval $[0,1]$ (i.e., by sampling from $U([0,1])$, line~\ref{algb:samplepdes}), and swapping the nodes $u,v$ if $p<1/2$ (line~\ref{algb:desswap}), before defining the DES $\swap{(u,v),(x,y)}{(u,x),(v,y)}$ in line~\ref{algb:defdef}. 
If the sampled DES is not a CDES, i.e., 
if after applying the DES to $G$ the resulting graph would have a different CDM, 
the algorithm discards the sampled DES and continues to the next iteration. 
Otherwise, if the sampled DES is a CDES, it evaluates the acceptance probability for the MH approach. 
This is done in lines~\ref{algb:rhostart}-\ref{algb:rhoend}. 
These operations are analogous to the standard algorithm to sample from the CM (Algorithm~3 from~\cite{fosdick2018configuring}). 
The key parameter computed by the algorithm is the proposal distribution ratio $\rho = \xi_H(G)/\xi_G(H)$; 
since the pair of edges $(u,v) , (x,y)$ is sampled uniformly at random from the multiset $E$ (i.e., all pairs of edges are equally likely to be sampled), such ratio mainly depends on the multiplicities of the edges $(u,v) , (x,y) , (u,x) , (v,y)$ involved in the DES. 
These are defined by the function $w_E$, where $w_E(u,v)$ is the number of occurrences of the edge $(u,v)$ in the multiset $E$. 
For instance, when the DES involves $4$ distinct nodes, 
the parameter $\rho$ is $\frac{(\omega_E(u,x)+1)(\omega_E(v,y)+1)}{\omega_E(u,v) \omega_E(x,y)}$ (line~\ref{algb:rhostart}), i.e.,
counting the number of equivalent DES that involve copies of the sampled pair of edges.
Following the MH approach, the move from the current state $G$ to the new state $H$ is accepted with probability $\rho \pi(H)/\pi(G)$ (line~\ref{algline:updateg} and~\ref{algline:acceptancef}). 
if the move is accepted, $G$ is replaced by $H$. 
We note that, when the target distribution $\pi$ is uniform over the space, as it is typical in applications,  
the acceptance probability is simply $\rho$. 

The following result states that the stationary distribution of random samples returned by \algname-B is unique and is equal to $\pi$.

\begin{theorem}
\label{thm:correctbaseline}
The Markov chain run by \algname-B has stationary distribution $\pi$. 
\end{theorem}

\textbf{\algname\ algorithm.}
We now introduce \algname, a refined approach to efficiently generate random multigraphs with prescribed CDM. 
The key insight in the design of \algname\ is a precise characterization of DES that can form changing CDES. 
We prove a \emph{necessary} condition for a pair of edges to realize a changing CDES. 
\begin{lemma}
\label{lemma:cdescond}
Any changing CDES involving the edges $(u,v) , (x,y)$
satisfies $\{ c(u) , c(v) \} = \{ c(x) , c(y) \}$.
\end{lemma}
Lemma~\ref{lemma:cdescond} specifies that all pairs of edges $(u,v) , (x,y)$ such that $\{ c(u) , c(v) \} \neq \{ c(x) , c(y) \}$ surely do \emph{not} form changing CDES, i.e., operations to move on the chain $\G$. 
A pair of edges $(u,v) , (x,y)$ that satisfies $\{ c(u) , c(v) \} = \{ c(x) , c(y) \}$ 
falls under two possible cases:
\begin{itemize}
\item $c(u) = c(v)$: the sampled edges are monochromatic and have the same color; both DES are CDES;
\item $c(u) \neq c(v)$: the sampled edges are bichromatic and have the same colors; only one DES is a CDES. Assuming $c(u) = c(x)$, only the DES $\swap{(u,v),(x,y)}{(u,y),(x,v)}$ is a CDES.
\end{itemize}
By leveraging these observations, 
we design 
\algname\ to only sample pairs of edges with
$\{ c(u) , c(v) \} = \{ c(x) , c(y) \}$,
which are the only candidates to be (possibly changing) CDES, thus proposing moves on the state graph $\G$ tailored to preserve the CDM. 

\algname\ is described in Algorithm~\ref{alg:ccm}. 
First, it samples uniformly at random an edge $(u,v)$ from $E$ (line~\ref{algccm:sampleuv}); 
given $\ell = c(u)$, and $r = c(v)$, it samples 
uniformly at random an edge $(x,y)$ from $E_{\ell , r} \setminus \{ (u,v) \}$ (line~\ref{algccm:samplexy}),
the multiset of edges of colors $\ell , r$. 
We may note that $E_{\ell , r} \setminus \{ (u,v) \}$ is empty if $|E_{\ell , r}| = 1$; we avoid this case with the following simple optimization. 
\algname\ samples $(u,v) \in E$ in line~\ref{algccm:sampleuv} if and only if $|E_{\ell , r}| \geq 2$: from Lemma~\ref{lemma:cdescond}, edges that do not satisfy this condition do not form changing CDES, and are removed from $E$ before the $t$ iterations of the algorithm with a simple preprocessing, and re-inserted at the end of the iterations. 
For simplicity and w.l.o.g., we assume that for the pair of edges $(u,v)$,$(x,y)$ sampled from 
$E_{\ell , r}$ it holds $c(u) = c(x) = \ell$, 
and 
$c(v) = c(y) = r$. 
Then, if $\ell = r$, \algname\ samples uniformly at random one of the two possible DES, which are both CDES; 
otherwise, when $\ell \neq r$, it deterministically choose the only DES that is a CDES, which is 
$\swap{(u,v),(x,y)}{(u,y),(x,v)}$. 
To do so, \algname\ swaps the nodes $u,v$ in line~\ref{algccm:cdeschange}
before defining the CDES (line~\ref{algccm:cdesdef}). 
Then, \algname\ computes the proposal distribution ratio $\rho = \xi_H(G)/\xi_G(H)$ (lines~\ref{algccm:ratiostart}-\ref{algccm:ratioend}), 
which depends on the multiplicities $w_E(u,v)$, $w_E(x,y)$, $w_E(u,x)$, $w_E(v,y)$, of the edges involved in the CDES, 
analogously to \algname-B. 
The proposed new state $H$, obtained by applying the sampled CDES to $G$ (line~\ref{algccm:applycdes}), 
is accepted with probability $\rho \pi(G)/\pi(H)$ (line~\ref{algccm:accept}),
following the MH approach.

\begin{algorithm}[t]
  \caption{\algname}\label{alg:ccm}
  \DontPrintSemicolon%
  \SetKwFor{RepTimes}{repeat}{times}{end}
  \SetKwRepeat{Do}{do}{while}
  \SetKw{Continue}{continue}
  \SetKw{Or}{or}
  \SetKw{And}{and}
  \KwIn{Observed multigraph $G \doteq (V, E, \mathcal{L}, c )$,
    distribution $\pi$ over $\mathcal{G}$, number of iterations $t$}
  \KwOut{Multigraph drawn from $\mathcal{G}$ according to $\pi$}
  \RepTimes{$t$}{
  $(V, E, \mathcal{L}, c ) \gets G$\;
      $(u,v) \gets$ sample from $U(E)$\; \label{algccm:sampleuv}
      $\ell , r \gets c(u) , c(v)$\;
      $(x,y) \gets$ sample from $U(E_{\ell,r} \setminus \{( u,v )\})$\; \label{algccm:samplexy}
      $p \gets$ sample from $U([0,1])$\;
      \lIf{$\ell \neq r$ or $p < 1/2$}{
      $(u,v) \gets (v,u)$ \label{algccm:cdeschange}
      }
      $\mathsf{cdes} \gets \swap{(u,v),(x,y)}{(u,x),(v,y)}$\; \label{algccm:cdesdef}
    \lIf{$|\{ u,v,x,y \}| = 4$}{
    $\rho \gets \frac{(\omega_E(u,x)+1)(\omega_E(v,y)+1)}{\omega_E(u,v) \omega_E(x,y)}$ \label{algccm:ratiostart}
    }
    \If{$|\{ u,v,x,y \}| = 3$}{
    \lIf{$u=v$ or $x=y$}{
    $\rho \gets \frac{(\omega_E(u,x)+1)(\omega_E(v,y)+1)}{2 \omega_E(u,v) \omega_E(x,y)}$
    }
    \lElse{
    $\rho \gets \frac{2(\omega_E(u,x)+1)(\omega_E(v,y)+1)}{ \omega_E(u,v) \omega_E(x,y)}$
    }
    }
    \If{$|\{ u,v,x,y \}| = 2$}{
        \lIf{only $(u,v)$ or $(x,y)$ is a self-loop}{ 
        \textbf{continue}
        }
        \If{both $(u,v)$ and $(x,y)$ are self-loops}{ 
        $\rho \gets \frac{(\omega_E(u,x)+2)(\omega_E(u,x)+1)}{ 4 \omega_E(u,u) \omega_E(x,x)}$\;
        }
        \lElse{
        $\rho \gets \frac{4(\omega_E(u,u)+1)(\omega_E(v,v)+1)}{ \omega_E(u,v)( \omega_E(u,v)-1)}$
        }
    }
    \lIf{$|\{ u,v,x,y \}| = 1$}{\textbf{continue} \label{algccm:ratioend} }
    $H \gets $ apply $\mathsf{cdes}$ to $G$\;\label{algccm:applycdes}
    $p \gets $ sample from $U([0,1])$\;
    \lIf{$p < \frac{\rho \pi(H)}{\pi(G)}$}{%
      $G \gets H$\label{algccm:accept}
    }
  }
  \Return{$G$}\;
\end{algorithm}

The following result implies the correctness of \algname. 
\begin{theorem}
\label{thm:correctalgo}
The Markov chain run by \algname\ has stationary distribution $\pi$. 
\end{theorem}
The theorem follows from the facts that \algname\ always samples pairs of edges that form CDES, that $\rho$ computes the proposal distribution ratio correctly, that the chain accepts the proposed moves with the MH approach, and that $\G$ is ergodic. 

While we have shown that both \algname\ and \algname-B 
converge to the stationary distribution $\pi$, 
we prove that \algname\ has a significant advantage w.r.t. 
\algname-B in terms of the mixing time, i.e., the number of iterations required to converge to the target distribution~$\pi$. 
This follows from the design of \algname, that directly samples CDES, instead of wasting iterations on operations that are surely not CDES, or non-changing ones. 
We quantify such advantage by comparing the entries of the \emph{transition matrices} of the two methods. 
Let $P$ and $P^B$ be, respectively, the transition matrices of the Markov chains run by \algname\ 
and \algname-B. 
Recall that the entry $P_{i,j}$ is the probability that 
the Markov chain moves from state $i$ to state $j$, 
and that $\sum_{j} P_{i,j} = 1$. 
The following result shows the existing relationship between the diagonal entries 
$P^B_{i,i} , P_{i,i}$ (i.e., the probabilities that the chains remain in the state $i$)
and 
off-diagonal entries 
$P^B_{i,j} , P_{i,j} , i \neq j$, (i.e., the probabilities that the chains move from the state $i$ to a different state $j$) of such matrices. 

\begin{lemma}
\label{lemma:ccmcomp}
Define $\theta$ as
\begin{align*}
\theta = \sum_{\ell \in \L} \frac{(|E_{\ell , \ell}|-1) |E_{\ell , \ell}|}{(|E|-1)|E|}
+ \sum_{\ell , r \in \L : \ell < r} 
\frac{(|E_{\ell , r}|-1) |E_{\ell , r}|}{2(|E|-1)|E|} .
\end{align*}
For any pair of states $i,j$, it holds
$P^B_{i,i} = 1- \theta + P_{i,i} \theta $, and 
$P^B_{i,j} = P_{i,j} \theta $.
\end{lemma}
The direct consequence of Lemma~\ref{lemma:ccmcomp}
is that all diagonal entries of the transition matrices satisfy
$P^B_{i,i} > P_{i,i}$, and for all off-diagonal entries 
it holds $P^B_{i,j} < P_{i,j}$ (since $\theta < 1$ when $|\L| \geq 2$). 
Consequently, 
\algname\ precedes \algname-B 
in Peskun’s order~\cite{peskun1973optimum}. 
This implies that \algname\  
has strictly smaller \emph{asymptotic variance} (Thm.~2 of~\cite{mira2001ordering}), 
since the eigenvalues of $P$ are smaller than $P^B$. 
A precedence in Peskun’s order is not sufficient to claim smaller mixing time, as it depends on the absolute spectral gap, i.e., the maximum \emph{absolute} eigenvalue of the transition matrix. 
However, a precedence in Peskun’s order implies smaller mixing time for transition matrices that are positive definite (i.e., all eigenvalues are positive, Thm.~6 of~\cite{mira2001ordering}). 
A very simple modification to \algname\ and \algname-B yields this condition.

Suppose that both \algname\ and \algname-B are made lazy, i.e., that at each iteration, they remain at the current state with probability $1/2$. 
It is well known that the transition matrix of a lazy chain is positive definite; it follows that \algname\ has strictly smaller mixing time than \algname-B, i.e., 
it converges faster to the stationary distribution $\pi$. 
We note that for running such a lazy variant of \algname\ (or \algname-B) it is not required to waste (approximately) half of the $t$ iterations; 
it is enough to run the chain for a random number of steps (a sample from the Binomial distribution with $2t$ trials and $1/2$ success probability, to perform the desired $t$ iterations on average). 

The key intuition behind Lemma~\ref{lemma:ccmcomp} is that 
$\theta$ represents the probability that a randomly sampled DES is either not a CDES, or a non-changing CDES, since it does not satisfy the conditions given in Lemma~\ref{lemma:cdescond}; 
therefore, it is expected that (at most) 
a fraction $\theta$ of iterations of \algname-B will propose valid moves on the state space. 
The value of $\theta$ depends on the number of colors and the relative size of the corresponding edge multisets, and is typically small in particular when the number of colors of $G$ grows. 
For instance, consider a graph with $|\L| = k$ colors with $| E_{\ell , r} | \approx | E_{\ell^\prime , r^\prime} |$ for all pairs of colors; 
in such a case, a simple calculation yields that $\theta = \BT{k^{-2}}$. 

\textbf{Running time analysis.} 
After a $\BO{ n + m }$ preprocessing, 
all operations performed by \algname\ and \algname-B at each iteration of the main loop require constant time (apart from querying and updating the dictionary of edge multiplicities $w_E$, taking average constant time). 
Therefore, the time and space complexities of the algorithms are, respectively, $\BO{ n + m + t }$ and $\BO{ n + m }$.

\section{Experimental evaluation}
\label{sec:experiments}
Our experimental evaluation has three goals: 
1)~we highlight the quantitative differences between random samples drawn from CM, \polaris, and \algname\ in terms of the obtained CDMs, i.e., the number of colored neighbors of each node on real-world networks and on their randomized versions; 
2)~we compare the performance of \algname\ with the baseline algorithm \algname-B, in terms of convergence and running time; 
3)~we apply \algname\ to evaluate the statistical significance of polarization measures on real-world vertex-colored graphs, and compare the results with state-of-the-art null models.

\textbf{Datasets.}
We consider 13 colored real-world networks from several domains, 
whose statistics are shown in Table~\ref{tab:graphs}. 
We briefly describe them.
Twitter~\cite{conover2011political}, 
Brexit, US-Elect, Abortion~\cite{garimella2017balancing}, 
Obamacare~\cite{garimella2017ebb}, 
Comb, and Guns~\cite{garimella2018political}  
are networks built from tweets, where vertices are users, the colors represent different opinions on controversial topics, and edges represent retweets between users. 
Com-Youtube and Com-DBLP~\cite{mislove2007socialnetworks,yang2012defining} 
are networks built from the Youtube social network, and from the DBLP Computer Science bibliography database. Edges represent, respectively, friendships between users and co-authorships between researchers;
the color of a node denotes if it belongs to any of the highest quality ground truth communities~\cite{yang2012defining}. 
Cite and Phy-Cit~\cite{preti2023maniacs} are citation networks,
where nodes represent publications and colors
denote Computer Science areas and the year of publication.
Trivago~\cite{chodrow2021generative} is a network 
where a node is an accommodation, and each edge connects accommodations visited by users in the same browsing session.
The color of a node denotes the accommodation's country. 
Walmart~\cite{amburg2020clustering} is a network where nodes are products, and edges connect products that were bought together.
The colors denote the products' departments.

\begin{table} 
  \caption{Statistics of the graphs considered in our experiments.  
  $n=|V|$ is the number of vertices, $m=|E|$ is the number of edges, $|\mathcal{L}|$ is the number of distinct nodes' colors, 
  $\hat{q} = \max_{\ell \in \L } | V_{ \ell , \ell } |/n$ is the maximum fraction of nodes of any color, 
  $\theta$ is the parameter defined in Lemma~\ref{lemma:ccmcomp}.
  }
\label{tab:graphs}
\center
  \begin{tabular}{lrrccc}
    \toprule
    $G$              & $n$      & $m$ & $|\mathcal{L}|$ & $\hat{q}$ & $\theta$   \\
    \midrule
    Cite & 3264 & 4611 & 6 & 0.21 & 0.118 \\
    Twitter & 22405 & 77920 & 3 & 0.51 & 0.365 \\
    Brexit & 22745 & 48830 & 2 & 0.50 & 0.503 \\
    US-Elect & 23832 & 845152 & 3 & 0.68 & 0.688 \\
    Phy-Cit & 30501 & 347268 & 6 & 0.13 & 0.014 \\
    Walmart & 88860 & 2267396 & 11 & 0.29 & 0.237 \\
    Trivago & 172738 & 1327092 & 160 & 0.09 & 0.049 \\
    Abortion & 279505 & 671144 & 2 & 0.50 & 0.524 \\
    Com-DBLP & 317080 & 1049866 & 2 & 0.71 & 0.321 \\
    Obamacare & 334617 & 1511670 & 2 & 0.50 & 0.598 \\
    Comb & 677753 & 6666018 & 2 & 0.50 & 0.550 \\
    Guns & 632659 & 7478993 & 2 & 0.50 & 0.423 \\
    Com-Youtube & 1134890  & 2987624 & 2 & 0.96 & 0.444 \\
  \bottomrule
\end{tabular}
\end{table}

\textbf{Experimental setup.} 
We implemented \algname\ and \algname-B in python. 
Our implementation, with data and scripts to reproduce the experiments, is available online.\footnote{https://github.com/leonardopellegrina/Sirius}
For CM and \polaris, we used the implementation from~\cite{preti2025polaris}. 
For all algorithms, we set the number of iterations $t=m \ln m$,
where $m$ is the number of edges of the graph, 
and sample from the uniform distribution over the respective graph spaces. 
We performed all experiments 
on a machine with 2.30 GHz Intel Xeon
CPU, 1 TB of RAM, on Ubuntu~22.04.

\textbf{Comparison with CM and \polaris.} 
In this first experiment we test for differences in terms of local assortativity of the nodes, i.e., measuring the fraction of same-color neighbors of each node. 
We compute these values from the observed, original networks, and compare them with their randomized counterpart obtained from CM and Polaris.
Doing so, we evaluate any difference in terms of the CDM between real-world networks and random samples from state-of-the-art null models.

\begin{figure}
\begin{subfigure}{.29\textwidth}
  \centering
  \includegraphics[width=\textwidth]{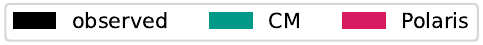}
\end{subfigure} \\
\begin{subfigure}{.35\textwidth}
  \centering
  \includegraphics[width=\textwidth]{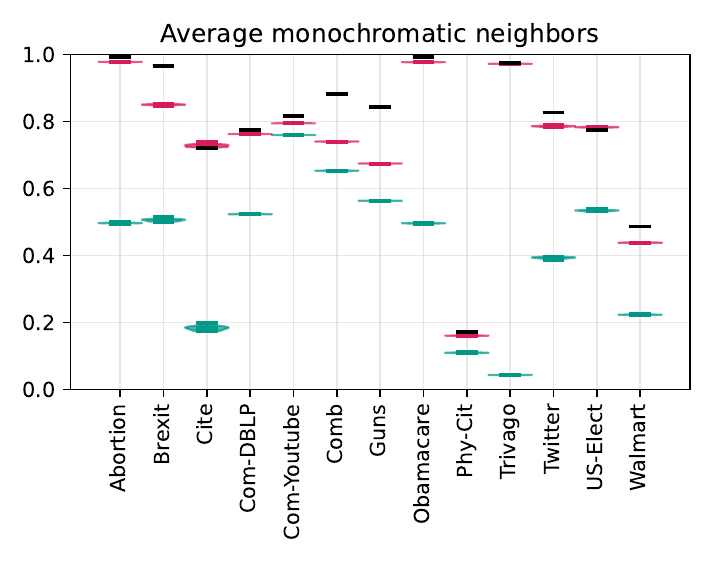} 
\end{subfigure}
\begin{subfigure}{.35\textwidth}
  \includegraphics[width=\textwidth]{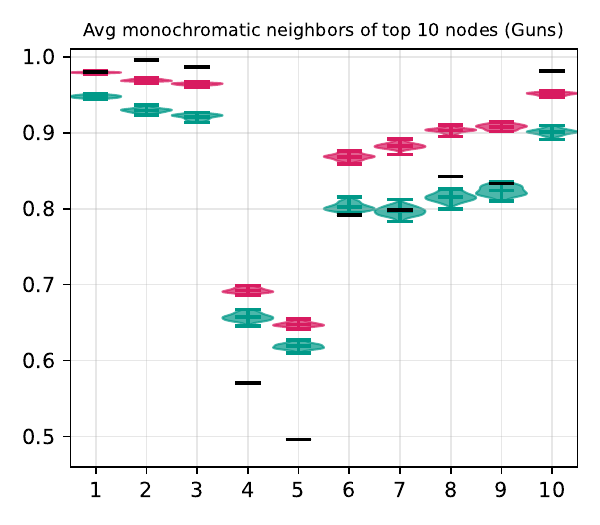}
\end{subfigure}
\caption{  $M$ and $M_v$ on observed vs random graphs. 
}
\Description{This figure shows two plots, reporting the values of $M$ and $M_v$ on observed vs random graphs.}
\label{fig:avgcolall}
\end{figure}

We define $M_v(G)$ as the fraction of neighbors of $v$ with the same color of $v$:
$M_v(G) = |N_G^{c(v)}(v)| / |N_G(v)|$. 
Then, we compute the average value of $M_v(G)$ over all nodes $v \in V$: $M(G) = \frac{1}{n} \sum_{v \in V} M_v(G)$. 
Intuitively, when $G$ represents a social network, 
$M_v$ quantifies the color assortativity of the node $v$,  
while $M$ its average over all nodes. 
We use $M$ as a proxy for comparing the CDM of two graphs; it is immediate to observe that for two graphs $G,H$ with the same CDM it holds
$M(G) = M(H)$, and $M_v(G) = M_v(H), \forall v \in V$. 
On the other hand, for two graphs $G,H$ with different CDMs it is reasonable to expect $M(G) \neq M(H)$. 

We compare, for each graph $G$ (Table~\ref{tab:graphs}), 
the value of $M(G)$ with $M(G^0_i)$, where $\{G^0_1 , G^0_2 , \dots , G^0_z\}$ are $z$ random graphs generated from the CM and Polaris null models. We use $z=100$ random samples. 
Note that we do not consider samples from the CCM as there would be no differences (since the CDM is preserved).

We show these results in Figure~\ref{fig:avgcolall}. 
We label with ``observed'' the value of $M(G)$ observed on real-world graphs, and with ``CM'' and ``Polaris'' the values from the null models. 
We note that the values of $M(G)$ are significantly different than the values $M(G^0_i)$ computed on random graphs from the CM. 
This may be expected, since CM completely ignores the colors of the nodes, and only preserves the degree sequence.
Interestingly, we also observe significant differences for random graphs drawn from the Polaris model, even if the quantitative gap is smaller in some cases. 
These findings confirm that only preserving the JCM, in addition to the degree sequence, yields random networks with significantly different CDMs, i.e., with a substantial change in the proportions of colored neighbors, even if averaged over all vertices of the graph. 

To inspect this observation in more detail, we also report the values of $M_v(G)$ for the $10$ nodes with highest degree of each graph.
Figure~\ref{fig:avgcolall} shows the results for the graph Guns 
(similar to other graphs, 
\ifextversion
shown in Fig.~\ref{fig:avgcolcmappendix}). 
\else
shown in the extended version~\cite{siriusextended}). 
\fi
Also in this case, we observe substantial differences between the values $M_v(G)$ observed on the original graphs, and the values $M_v(G^0_i)$ from the randomized samples; 
for both CM and Polaris, such values can be either higher or smaller than observed. 
For instance, in randomized versions of a social network, a given node may be overexposed, or underexposed, to like-minded individuals, compared to the observed network; 
depending on the goals and assumptions of the test, not controlling for this effect may introduce a bias in the statistical analysis. 
From these findings we conclude that random graphs from such models do not preserve the CDM, and thus output samples with substantial differences in terms of the colors of the neighbors of the nodes. 
In other words, preserving the \emph{global} color assortativity does not yield graphs with prescribed \emph{local} color assortativity. 
Therefore, existing models do not allow evaluating the effect of the CDM to downstream statistical analyses, such as the evaluation of the significance of a given polarization measure.

\textbf{Performance of \algname\ and \algname-B.}
In this section we test \algname\ and \algname-B in terms of convergence and running time. 

First, we compare \algname\ with \algname-B in terms of convergence. 
For both methods we use the same number of iterations $t = m \ln m$, where $m$ is the number of edges of the graph $G$.
Every $t/100$ iterations, we store the obtained graph and compute its degree assortativity, a commonly employed diagnostic measure of convergence~\cite{dutta2025sampling}. 
We repeat these tests $10$ times for each graph. 
Figure~\ref{fig:convergenceccm} shows the results for the graphs Phy-Cit and Trivago 
\ifextversion
(others in Fig.~\ref{fig:convergenceccmappx}). 
\else
(other graphs in the extended version~\cite{siriusextended}). 
\fi
The plots display the assortativity of the $10$ independent runs. 
We observe that \algname\ quickly reaches a plateau, an empirical sign of convergence to the target distribution $\pi$. 
Instead, \algname-B is much slower to converge, and does not reach the same plateau within the $t$ iterations. 
These findings show that \algname\ is drastically more effective in exploring the state space 
of multigraphs with fixed~CDM.

\begin{figure}[ht]
\begin{subfigure}{.18\textwidth}
  \centering
  \includegraphics[width=\textwidth]{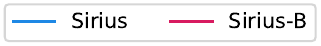}
\end{subfigure} \\
\begin{subfigure}{.235\textwidth}
  \includegraphics[width=\textwidth]{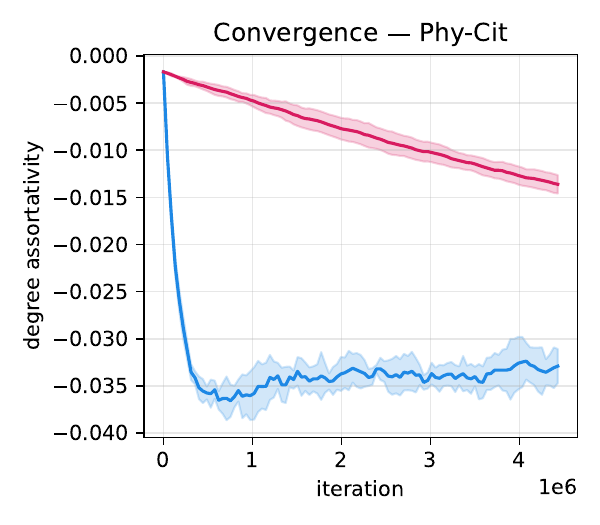}
\end{subfigure}
\begin{subfigure}{.235\textwidth}
  \centering
  \includegraphics[width=\textwidth]{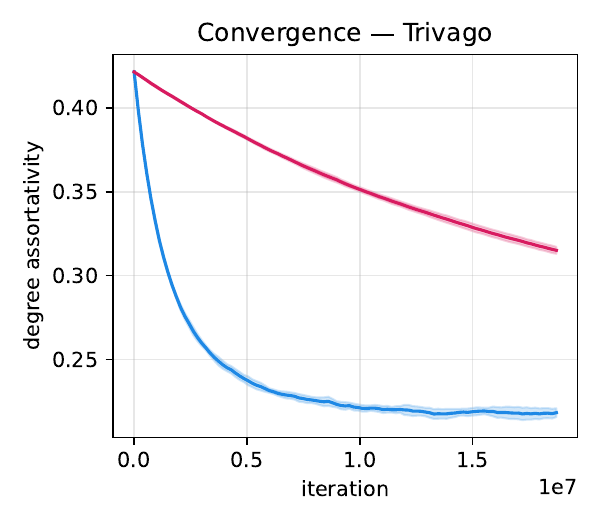} 
\end{subfigure}
\caption{  Degree assortativity of \algname\ and \algname-B.
}
\Description{This figure shows two plots, reporting the degree assortativity values of \algname\ and \algname-B on two graphs (Phy-Cit and Trivago).}
\label{fig:convergenceccm}
\end{figure}

\begin{figure}[ht]
\begin{subfigure}{.5\textwidth}
  \centering
  \includegraphics[width=.9\textwidth]{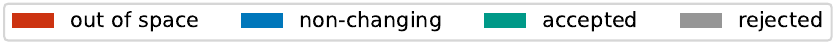}
\end{subfigure}  \\
\begin{subfigure}{.156\textwidth}
  \centering
  \includegraphics[width=\textwidth]{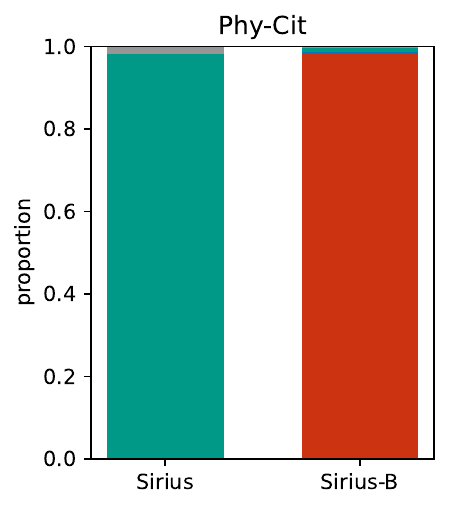}
\end{subfigure}
\begin{subfigure}{.156\textwidth}
  \includegraphics[width=\textwidth]{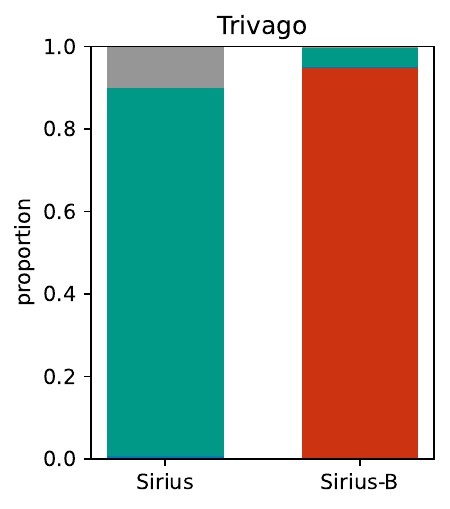}
\end{subfigure}
\begin{subfigure}{.156\textwidth}
  \centering
  \includegraphics[width=\textwidth]{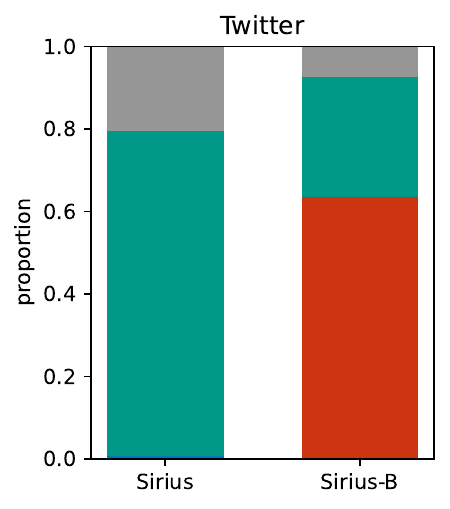} 
\end{subfigure}
\caption{  Iteration outcomes of \algname\ and \algname-B.
}
\Description{This figure shows three plots, reporting the fraction of iteration outcomes values of \algname\ and \algname-B on three graphs (Phy-Cit, Trivago, and Twitter).}
\label{fig:iterstatsmain}
\end{figure}

To understand the reasons for this gap, we report in Figure~\ref{fig:iterstatsmain} the statistics of the outcomes of the $t$ iterations performed by both methods on three graphs 
\ifextversion
(others in Fig.~\ref{fig:iterstatsappx}).
\else
(others in~\cite{siriusextended}).
\fi
The plots show the proportion of iterations resulting in ``out of space'' (when a sampled DES is not a CDES), 
``non-changing'' (when a sampled DES is non-changing CDES),
``accepted'' (when a changing CDES is accepted), 
and 
``rejected'' (when a changing CDES is rejected). 
(Note that ``non-changing'' values are extremely small and non visible in the plots.)
We observe a striking difference between the behavior of \algname\ and \algname-B: a large proportion of iterations of \algname-B is out of space, denoting that most of sampled DES are not CDES. 
Instead, as expected from its design, \algname\ only samples CDES. 
Moreover, a high fraction of proposed moves is accepted, 
a sign that it effectively moves on the chain $\G$. 
Interestingly, we observe that the fraction of iterations of \algname-B that results in out of space is extremely close to $\theta$ (reported in Table~\ref{tab:graphs} for each graph). 
This aligns with our theoretical analysis of the transition matrices $P$ and $P^B$ (Lemma~\ref{lemma:ccmcomp}). 
\ifextversion
In Figure~\ref{fig:convergenceccmappxmoreit} we also show 
\else
In~\cite{siriusextended} we also show 
\fi
the degree assortativity of \algname\ and \algname-B using $t = 10^2 m \ln m$ iterations; we observe that \algname\ reaches a stable plateau for all graphs, a robust sign of convergence. 

We now compare the methods in terms of running time. 
We run \algname\ with $t$ iterations, and run \algname-B until it samples $t$ CDES (i.e., until the total number of ``accepted'', ``rejected'', or ``non-changing'' iterations is $t$), to make sure that the random networks reported by both methods have an equal degree of convergence to $\pi$. 
\ifextversion
Figure~\ref{fig:runtimesall} shows the measured running time. 
\fi
From the plot 
\ifextversion
\else
(shown in the extended version~\cite{siriusextended}) 
\fi 
we observe that, even if \algname-B can rapidly check if a sampled DES is not a CDES, it is still significantly slower than \algname\ to sample and complete $t$ valid iterations, by at least a factor $\approx 2$, up to almost $2$ orders of magnitude. 
Finally, we found that the running time of \algname\ is comparable to CM and Polaris 
\ifextversion
(Figure~\ref{fig:runtimesall}). 
\else
(Figure in~\cite{siriusextended}). 
\fi

We conclude that \algname\ is much more efficient that the baseline method \algname-B, and is an effective tool to sample random networks with prescribed CDM.

\begin{figure}[ht]
\begin{subfigure}{.345\textwidth}
  \centering
  \includegraphics[width=\textwidth]{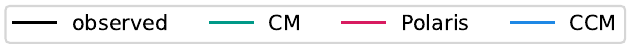}
\end{subfigure} \\
\begin{subfigure}{.345\textwidth}
  \centering
  \includegraphics[width=\textwidth]{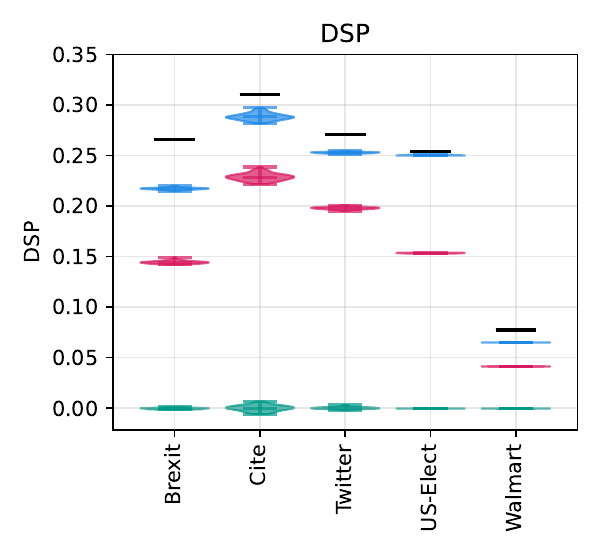} 
\end{subfigure}
\begin{subfigure}{.47\textwidth}
  \includegraphics[width=\textwidth]{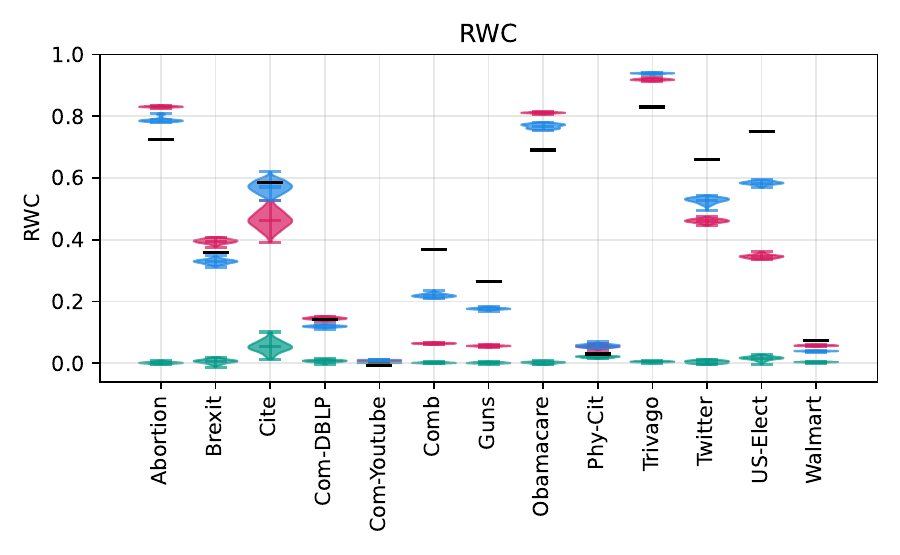}
\end{subfigure}
\caption{Polarization scores on observed vs random graphs. 
}
\Description{This figure shows two plots, reporting the polarization scores (DSP and RWC) on observed vs random graphs.}
\label{fig:polarization}
\end{figure}

\textbf{Evaluation of Polarization Measures.}
In this final experiment, we apply \algname\ to the task of assessing the statistical significance of polarization measures for colored multigraphs. 
While an in-depth investigation of structural polarization is beyond the scope or our work, 
our main goal is to assess the impact of preserving the CDM in a realistic scenario, i.e., to quantify the differences between random networks generated from null models with different characteristics in a practical network analysis task. 
For this experiment we increase the iterations to $t=10^2 m \ln m$ to avoid any artifact potentially due to convergence (however, we note that we always obtained the same results when using the default number of iterations). 

We focus on two measures: 
the popular Random Walk Controversy (RWC) score~\cite{garimella2018quantifying} 
and the recent DSP score~\cite{preti2026dsp}. 
The former quantifies the difference between the probability that a random walk (either unbounded or with restart) starts from a random user from one group, and reaches an influential user (e.g., a node with high degree) of the same group, and the probability that the random walk starts and ends on an influential node of a different group. 
The RWC score is close to one when the probability of crossing sides is low, and close
to zero when the probability of crossing sides is comparable to that of staying on the same side. 
Similarly to previous work~\cite{garimella2018quantifying,preti2026dsp}, we choose the set of $10$ highest-degree nodes of each community as the set of influential nodes. 
The DSP score uses a similar diffusion-based model, but 
with a refined distribution over random walks, taking into account unbalanced communities (i.e., groups of non-equal size); 
moreover, it does not rely on specifying a set of influential users of each community, but instead considers the probability that a random walk reaches any node (from the same or different community). 
DSP was shown to be less prone to bias or misinterpretation~\cite{preti2026dsp}. 
Its range is $(-1/2 , 1/2)$, taking the maximum value when 
diffusions from a random node of a community 
only reach
vertices of the same community, while the minimum is achieved when 
they only reach
vertices of the opposite community. 
Since the current implementations only support two communities, 
for graphs with more than~$2$ colors we define the first group using the color with the most vertices, and include in the second group all other colors. For both measures we used random walks with restart, with default restart probability ($0.15$). 
We only report the results for the DSP measure for the $5$ smallest networks, due to high running times (as computing DSP has $\BO{nm}$ time complexity). 
In Figure~\ref{fig:polarization} we show the scores from the observed graphs and from $100$ random networks for each null.

We first discuss the results for DSP. 
We see that this score is always close to $0$ for random graphs sampled from the CM. 
This suggests that for all considered graphs, the degree sequence alone does not explain the observed polarization. 
Random samples drawn with Polaris show higher values of the DSP score compared to CM, however, such quantity still is considerably higher in all the observed graph compared to the random samples. 
Therefore, the observed score is not explained by the JCM of the graphs.
Regarding samples from CCM, we see that in all cases the score is substantially higher compared to Polaris; 
on the other hand, the observed scores for almost all graphs are still significantly higher in the original networks.
This suggests that, in such cases, the observed level of polarization is not explained by the CDM alone, and that other factors have a key role. 
For one graph (US-Elect), the observed DSP score is almost the same of randomized graphs from the CCM; 
in this case, the CDM is a characteristic of the graph that most likely explains the observed structural polarization, i.e., affecting the probability that users from each community interact with, or is exposed to, content originating from nodes of the opposite community. 
Therefore, for the DSP score the insights that a network analyst may gather may substantially differ depending on the considered null model. 

We now discuss the results for the RWC score. 
Similarly to DSP, we observe that its value is in all cases close to $0$ under CM. 
Interestingly, the RWC may be either higher, or smaller, for Polaris compared to CCM. 
In general, the score under CCM is typically closer to the score obtained form the observed graph; 
this may be expected as CCM preserves additional structure and thus reasonably outputs graphs with characteristics more similar to real-world graphs. 
On graphs from social networks, we observe that RWC is significantly higher than CCM in several cases (Comb, Guns, Twitter, US-Elect), and significantly smaller in others (Abortion, Obamacare), suggesting that preserving the CDM does not explain the observed score in such cases. 
Other factors may decrease, or increase, the probability of reaching an influential user from the opposite community. 
In one graph (Cite) the observed RWC score is not explained by the observed JCM, while it is explained by the CDM, offering new insights on the network's structure. 

We conclude that \algname\ is an effective tool to investigate nuanced properties of colored multigraphs, 
such as the role of the observed CDM to a given phenomenon, which cannot be assessed using state-of-the-art graph null models.

\section{Conclusions}
In this work we introduce the Colored Configuration Model (CCM), 
an expressive null model preserving the Colored Degree Matrix (CDM).
Motivated by the study of polarization on social networks, the CDM quantifies the color assortativity of all nodes, encoding rich structure of the observed network. 
We develop \algname, an efficient algorithm  
to sample random graphs from the CCM, 
representing a new tool for future investigation of complex phenomena with vertex-colored graphs, while evaluating the effect of the local color assortativity observed in real-world networks.

\begin{acks}
This work is supported by 
the STARS@UNIPD 2025 program, 
project "AtHeNA: Algorithms for Heterogeneous Network Analysis", 
with the support of the University of Padova and Fondazione Cassa di Risparmio di Padova e Rovigo,
and by 
the Italian Ministry of University and Research (MUR), project “National Center for HPC, Big Data, and Quantum Computing” CN00000013. 
\end{acks}

\clearpage
\newpage

\bibliographystyle{ACM-Reference-Format}
\balance
\bibliography{bibliography}

\appendix

\section{Appendix}
\ifextversion
This appendix presents the proofs of our theoretical derivations, and experimental results that could not fit in the main manuscript due to space constraints. 
\else
This appendix presents the proofs of our theoretical derivations
that could not fit in the main manuscript due to space constraints. 
\fi

\begin{algorithm}
  \small
  \caption{\algname-B}\label{alg:baseline}
  \DontPrintSemicolon%
  \SetKwFor{RepTimes}{repeat}{times}{end}
  \SetKwRepeat{Do}{do}{while}
  \SetKw{Continue}{continue}
  \SetKw{Or}{or}
  \SetKw{And}{and}
  \KwIn{Observed multigraph $G \doteq (V, E, \mathcal{L}, c )$,
    distribution $\pi$ over $\mathcal{G}$, number of iterations $t$}
  \KwOut{Multigraph drawn from $\mathcal{G}$ according to $\pi$}
  \RepTimes{$t$}{
  $(V, E, \mathcal{L}, c ) \gets G$\;
    \Do{$\mathsf{des}$ is not a CDES}{
      $(u,v) \gets$ sample from $U(E)$\; \label{algb:sampleone}
      $(x,y) \gets$ sample from $U(E \setminus \{( u,v )\})$\; \label{algb:sampletwo}
      $p \gets$ sample from $U([0,1])$\; \label{algb:samplepdes}
      \lIf{$p < 1/2$}{
      $(u,v) \gets (v,u)$ \label{algb:desswap}
      }
      $\mathsf{des} \gets \swap{(u,v),(x,y)}{(u,x),(v,y)}$\; \label{algb:defdef}
    }
    \lIf{$|\{ u,v,x,y \}| = 4$}{
    $\rho \gets \frac{(\omega_E(u,x)+1)(\omega_E(v,y)+1)}{\omega_E(u,v) \omega_E(x,y)}$
    } \label{algb:rhostart}
    \If{$|\{ u,v,x,y \}| = 3$}{
    \lIf{$u=v$ or $x=y$}{
    $\rho \gets \frac{(\omega_E(u,x)+1)(\omega_E(v,y)+1)}{2 \omega_E(u,v) \omega_E(x,y)}$
    }
    \lElse{
    $\rho \gets \frac{2(\omega_E(u,x)+1)(\omega_E(v,y)+1)}{ \omega_E(u,v) \omega_E(x,y)}$
    }
    }
    \If{$|\{ u,v,x,y \}| = 2$}{
        \lIf{only $(u,v)$ or $(x,y)$ is a self-loop}{ 
        \textbf{continue}
        }
        \If{both $(u,v)$ and $(x,y)$ are self-loops}{ 
        $\rho \gets \frac{(\omega_E(u,x)+2)(\omega_E(u,x)+1)}{ 4 \omega_E(u,u) \omega_E(x,x)}$\;
        }
        \lElse{
        $\rho \gets \frac{4(\omega_E(u,u)+1)(\omega_E(v,v)+1)}{ \omega_E(u,v)( \omega_E(u,v)-1)}$
        }
    }
    \lIf{$|\{ u,v,x,y \}| = 1$}{\textbf{continue}} \label{algb:rhoend}
    $H \gets $ apply $\mathsf{des}$ to $G$\;\label{algline:updateg}
    $r \gets $ sample from $U([0,1])$\;
    \lIf{$r < \frac{\rho \pi(H)}{\pi(G)}$}{%
      $G \gets H$\label{algline:acceptancef}
    }
  }
  \Return{$G$}\;
\end{algorithm}

\begin{proof}[Proof of Theorem \ref{thm:irreducible}]
Let $H_1 = ( V , E_1 , \L , c )$ and $H_2 = ( V , E_2 , \L , c )$ 
be an arbitrary pair of states of $\G$. 
We show that it is possible to transform $H_1$ into $H_2$ with a finite sequence of CDES.
For $j \in \{ 1 , 2 \}$ and any pair of colors $\ell , r$, 
define $E_{j , \ell , r}$ as the multiset of edges that are adjacent to nodes with colors $\ell , r$, that is
$E_{j , \ell , r} = \{\!\!\{ (u , v) \in E_j : \{ c(u) , c(v) \} = \{ \ell , r \} \}\!\!\}$. 
Then, 
define $V_{\ell , r} = \{ v : c(v) = \ell \vee c(v) = r \}$ as the set of nodes that have color $\ell$ or $r$. 
Note that $E_{j , \ell , r}$ is a partition of the multiset $E_j$, 
and that $\bigcup_{ \ell , r } V_{\ell , r} = V$. 
Then, let $H_{ j , \ell , r }$ be the subgraph of $H_j$ defined as $H_{ j , \ell , r } = (V_{\ell , r} , E_{j , \ell , r} , \L , c )$. 
We note that, for any $\ell \in \L$,
$H_{ j , \ell , \ell }$ is the multigraph formed by all nodes of color $\ell$,
and all monochromatic edges of color $\ell$;
for any $\ell \neq r$, $H_{ j , \ell , r }$ is the bipartite multigraph formed by all nodes of color $\ell$ or $r$,
and all bichromatic edges of colors $\ell , r$. 
It is immediate to observe that $H_1 = H_2$ if and only if 
$H_{ 1 , \ell , r } = H_{ 2 , \ell , r }$, for all $\ell , r \in \L$.
Moreover, from Lemma~\ref{lemma:cdescond}, any CDES on $H_j$ only affects a pair of edges 
from exactly one subgraph $H_{ j , \ell , r }$, for some $\ell , r$. 

It is well known 
that the state graph of multigraphs with nodes that have the same color is strongly connected (e.g., Lemma~2.14 in~\cite{fosdick2018configuring}); 
therefore, there always exists a finite sequence of CDES 
connecting $H_{ 1 , \ell , \ell }$ to $H_{ 2 , \ell , \ell }$, for any $\ell \in \L$. 
Given an arbitrary ordering of the colors of $\L$,
by repeatedly applying such sequences of CDES we modify 
$H_{ 1 , \ell , \ell }$ to be equal to $H_{ 2 , \ell , \ell }$, for each $\ell \in \L$. 

The same holds for the state graph of bipartite multigraphs; 
a DES that preserves the bipartition and the degree sequence of the nodes is a CDES. Therefore, 
there always exists a finite sequence of CDES 
connecting $H_{ 1 , \ell , r }$ to $H_{ 2 , \ell , r }$, for any $\ell,r \in \L$ with $\ell \neq r$. 
It is sufficient to repeatedly apply such sequences of CDES to transform 
$H_{ 1 , \ell , r }$ to be equal to $H_{ 2 , \ell , r }$, for each $\ell , r \in \L , \ell \neq r$. 

After applying these sequences of CDES to $H_1$, the resulting multigraph is equal to $H_2$. Since all CDES are reversible, the state graph $\G$ is strongly connected.
\end{proof}

\begin{proof}[Proof of Theorem \ref{thm:aperiodic}]
Suppose the first condition holds. 
Then, for all states $H \in \mathcal{V}$, there exist two monochromatic edges with color $\ell$. 
Denote such edges $(u,v) , (x,y)$ in $H$. 
Depending on the cardinality of $\{ u , v , x , y \}$, there are four possible cases:

$|\{ u , v , x , y \}| = 1$: then both edges are self-loops; both CDES involving $(u,v) , (x,y)$ form self-loops for $H$, thus $\G$ is aperiodic. 

$|\{ u , v , x , y \}| = 2$: if only one of $(u,v)$ or $(x,y)$ is a self-loop, then both CDES involving $(u,v) , (x,y)$ form self-loops for $H$.
If both $(u,v)$ and $(x,y)$ are self-loops, there is a cycle of length $3$ starting from $H$: 
$\swap{(u,v),(x,y)}{(u,y),(x,v)}$, $\swap{(u,y),(x,v)}{(u,x),(v,y)}$, $\swap{(u,x),(v,y)}{(u,v),(x,y)}$; 
moreover, 
since all CDES are reversible, there is also a cycle of length $2$. 
If both $(u,v)$ and $(x,y)$ are not self-loops, than 
one CDES is a self-loop for $H$.
Therefore, $\G$ is aperiodic. 

$|\{ u , v , x , y \}| = 3$: 
if both edges are not self-loops, then one CDES is a self-loop for $H$; 
else, if one of the two edges is a self-loop (assume it is $(x,y)$), then
there exists a cycle of length $3$ starting from $H$: 
$\swap{(u,v),(x,y)}{(u,x),(v,y)}$, $\swap{(u,x),(v,y)}{(v,x),(u,y)}$, $\swap{(v,x),(u,y)}{(u,v),(x,y)}$. 
Since all CDES are reversible, there is also a cycle of length $2$. 
Therefore, $\G$ is aperiodic. 

$|\{ u , v , x , y \}| = 4$: 
then, 
there exists a cycle of length $3$ starting from $H$: 
$\swap{(u,v),(x,y)}{(u,x),(v,y)}$, $\swap{(u,x),(v,y)}{(v,x),(u,y)}$, $\swap{(v,x),(u,y)}{(u,v),(x,y)}$. 
As before, 
since all CDES are reversible, there is also a cycle of length $2$. 
Therefore, $\G$ is aperiodic. 

We now suppose that the second condition holds. 
Such condition implies that in $G$ there must exist a pair of colors $\ell \neq r$ such that there exist two bichromatic edges with colors $\ell,r$ that are incident to a node $v$ with $c(v) = r$.
Note that, if such condition holds for $G$, then it also holds for all states $H \in \mathcal{V}$. 
For any $H \in \mathcal{V}$, denote such edges $(v,x)$ and $(v , y)$, assuming w.l.o.g. that $c(v) = r$ and $c(x) = c(y) = \ell$. 
We observe that the only CDES 
$\swap{(v,x),(v,y)}{(v,y),(v,x)}$ is non-changing, thus is a self-loop for $H$;
therefore, $\G$ is aperiodic. 
\end{proof}

\begin{proof}[Proof of Theorem \ref{thm:correctbaseline}]
The proof follows from the irreducibility and aperiodicity properties of the Markov chain $\G$ explored by \algname-B, proved by Theorems~\ref{thm:irreducible} and~\ref{thm:aperiodic}, 
and by 
the analysis of Algorithm~3 of~\cite{fosdick2018configuring}: 
all the sampled pairs of distinct edges are chosen uniformly at random from $E$,   
all CDES are DES, 
and 
the graph space given in input to Algorithm~3~\cite{fosdick2018configuring} is the set $\V$ of all multigraphs with the same CDM of $G$. 
\end{proof}

\begin{proof}[Proof of Lemma \ref{lemma:cdescond}]
We prove the result by considering all changing DES involving $(u,v) , (x,y)$, showing that in all cases, if they are CDES then it must hold $\{ c(u) , c(v) \} = \{ c(x) , c(y) \}$. 
There are $4$ possible cases, depending on the cardinality of the set $\{ u , v , x , y \}$:

$|\{ u , v , x , y \}| = 4$: in such a case, both DES are changing. 
Consider the DES $\swap{(u,v) , (x,y)}{(u,y),(x,v)}$; it is a CDES if and only if $c(v) = c(y)$, and $c(u)=c(x)$; this holds when $\{ c(u) , c(v) \} = \{ c(x) , c(y) \}$. 
Consider the DES $\swap{(u,v) , (x,y)}{(u,x),(y,v)}$; it is a CDES if and only if $c(v) = c(x)$, and $c(u)=c(y)$; this holds when $\{ c(u) , c(v) \} = \{ c(x) , c(y) \}$. 

$|\{ u , v , x , y \}| = 3$: either one edge is a self-loop, or both are not self-loops. W.l.o.g., in the first case assume $x=y$ (thus $c(x) = c(y)$).
Both DES are changing and are equivalent, and are CDES if and only if $c(x)=c(y)=c(u)=c(v)$, that holds when $\{ c(u) , c(v) \} = \{ c(x) , c(y) \}$. 
Consider the case where both edges are not self-loops; assume w.l.o.g. $u=x$ (thus $c(u)=c(x)$). Only the DES $\swap{(u,v) , (x,y)}{(u,x),(y,v)}$ is changing, and is a CDES if and only if $c(x)=c(y)=c(u)=c(v)$, that holds when $\{ c(u) , c(v) \} = \{ c(x) , c(y) \}$. 

$|\{ u , v , x , y \}| = 2$: the two edges are either both self-loops, or none is a self-loop; otherwise, there is no changing DES. 
Consider the case where both edges are not self-loops. 
Then, w.l.o.g. assume $u=x$ and $v=y$. Only the DES
$\swap{(u,v) , (x,y)}{(u,x),(y,v)}$
is changing; it is a CDES if and only if 
$c(x)=c(y)=c(u)=c(v)$, that holds when $\{ c(u) , c(v) \} = \{ c(x) , c(y) \}$. 

$|\{ u , v , x , y \}| = 1$: no DES is changing. 

Therefore, we have proved that all changing CDES involving the pair of edges $(u,v) , (x,y)$ satisfy the condition $\{ c(u) , c(v) \} = \{ c(x) , c(y) \}$, obtaining the statement. 
\end{proof}

\begin{proof}[Proof of Lemma \ref{lemma:ccmcomp}]
Let $X(t)$ and $X^B(t)$ be, respectively, random variables that denote the state assumed by the Markov chains run by \algname\ and \algname-B at a given iteration $t$. 
For any $i,j$, define the events $A_{i,j}$ and $A^B_{i,j}$, respectively, 
that the Markov chains run by \algname\ and \algname-B moves from the state $i$ to $j$. 
Note that $j=i$ corresponds to staying in the state $i$. 
By definition of such events, it holds
$P_{i,j} = \Pr( A_{i,j} | X(t)=i )$ and $P^B_{i,j} = \Pr( A^B_{i,j} | X^B(t)=i )$.
At any iteration $t$ of \algname-B, let $e_1 = (u,v)$ and $e_2 = (x,y)$ be the random variables modeling the two distinct edges sampled by the algorithm (in lines~\ref{algb:sampleone}-\ref{algb:sampletwo}), and $p$ be the random variable drawn from $U([0,1])$ in line~\ref{algb:samplepdes}. 
Define the events $C_{\ell , r}$ and $Z$ as 
\begin{align*}
C_{\ell , r} &= \text{``} \{ c(u) , c(v) \} = \{ c(x) , c(y) \} = \{ \ell , r \} \text{''} \\ 
Z &= \text{``} C_{\ell , r} \wedge \left( (  \ell = r ) 
 \vee (  \ell \neq r \wedge p < 1/2  ) \right) \text{''} .
\end{align*}
Note that $C_{\ell , r}$ is the event that the pair of sampled edges has the same set $\{ \ell , r \}$ of colors (both monochromatic when $\ell = r$, both bichromatic when $\ell \neq r$). 
The event $Z$ is true when $C_{\ell , r}$ holds and $\ell = r$, or when 
$C_{\ell , r}$ holds, $\ell \neq r$, and $p<1/2$. 
The event $Z$ models the fact that the sampled pair of edges by \algname-B has the same properties of a pair of edges sampled by \algname, i.e., that
happens when $C_{\ell , r}$ holds and, when the edges are bichromatic, the right CDES is sampled by \algname-B (using $p$). 
We observe that the probability that $Z$ holds, over the values of $e_1,e_2,p$, is $\Pr(Z | X^B(t)=i ) = \Pr(Z) = \theta$. 
From the law of total probability it holds
\begin{align*}
P^B_{i,i} = \Pr( A^B_{i,i} | X^B(t)=i , Z ) \Pr(Z) + \Pr( A^B_{i,i} | X^B(t)=i , \bar{Z} ) \Pr(\bar{Z}).
\end{align*}
Note that, when $Z$ is false, from Lemma~\ref{lemma:cdescond} we obtain that the DES sampled by \algname-B is either not a CDES, or is a non-changing CDES. 
Therefore, it holds $\Pr( A^B_{i,i} | X^B(t)=i , \bar{Z} ) = 1$. 
Otherwise, when $Z$ holds, the CDES sampled by \algname-B has the same probability to be rejected, or to be a non-changing CDES, than \algname, i.e., it holds
$\Pr( A^B_{i,i} | X^B(t)=i , Z ) = \Pr( A_{i,i} | X(t)=i ) = P_{i,i}$. Therefore, 
\begin{align*}
P^B_{i,i} 
&= \Pr( A^B_{i,i} | X^B(t)=i , Z ) \Pr(Z) + \Pr( A^B_{i,i} | X^B(t)=i , \bar{Z} ) \Pr(\bar{Z}) \\
&= P_{i,i} \theta + 1-\theta .
\end{align*}
Similarly, for $i \neq j$, we have 
\begin{align*}
P^B_{i,j} = \Pr( A^B_{i,j} | X^B(t)=i , Z ) \Pr(Z) + \Pr( A^B_{i,j} | X^B(t)=i , \bar{Z} ) \Pr(\bar{Z}) ,
\end{align*}
where $\Pr( A^B_{i,j} | X^B(t)=i , \bar{Z} ) = 0$, 
and $\Pr( A^B_{i,j} | X^B(t)=i , Z ) = \Pr( A_{i,j} | X(t)=i ) = P_{i,j}$. 
Therefore, 
\begin{align*}
P^B_{i,j} &= \Pr( A^B_{i,j} | X^B(t)=i , Z ) \Pr(Z) + \Pr( A^B_{i,j} | X^B(t)=i , \bar{Z} ) \Pr(\bar{Z})  \\
&= P_{i,j} \theta ,
\end{align*}
obtaining the statement.
\end{proof}

\begin{proof}[Proof of Theorem \ref{thm:correctalgo}]
First, all pairs of sampled edges $(u,v)$, $(x,y)$ by the algorithm satisfy $\{ c(u) , c(v) \} = \{ c(x) , c(y) \}$; from Lemma~\ref{lemma:cdescond}, these edges form a CDES. 
To show that the algorithm follows the MH approach, we prove that the variable $\rho$ is set correctly to proposal distribution ratio $\xi_H(G)/\xi_G(H)$. 
It is immediate to observe that each pair of edges $(u,v)$, $(x,y)$ with colors 
$\{ c(u) , c(v) \} = \{ c(x) , c(y) \} = \{ \ell , r \}$ is sampled by \algname\ with probability $(|E| (|E_{\ell , r}|))^{-1}$. 
All CDES are reversible, and the edges $(u,x)$, $(v,y)$ obtained after applying the CDES to $G$ will also belong to $E_{\ell , r}$, thus are also sampled with probability $(|E| (|E_{\ell , r}|))^{-1}$. 
This implies that $\rho = \xi_H(G)/\xi_G(H)$ only depends on the multiplicities of the edges involved in the sampled CDES from $G$ to~$H$.  
Define $V_{\ell , r} = \{ v : c(v) = \ell \vee c(v) = r \}$ as the set of nodes that have color $\ell$ or $r$. 
We consider two cases:

$c(u) = c(v) = \ell$: in this case, both edges are monochromatic, and have the same color $\ell$. 
The sampled CDES is equivalent to a DES performed on the subgraph $G_{\ell,\ell} = (V_{\ell , \ell} , E_{\ell , \ell} , \L , c)$, i.e., a multigraph where all nodes have the same color. The variable $\rho$ is then set correctly, following the analysis of Algorithm~3 of~\cite{fosdick2018configuring}.

$c(u) = c(y) = \ell$, $c(v) = c(x) = r$, $\ell \neq r$: in this case, both edges are bichromatic, with colors $\ell , r$. 
When $|\{ u , v , x , y \}| \leq 3$, the sampled CDES is non-changing; regardless of the value of $\rho$, the proposed state $H = G$, thus the iteration perfomed by \algname\ is correct. 
When $|\{ u , v , x , y \}| = 4$, the CDES is changing; the variable $\rho$ is set to the correct value 
$\rho = \frac{(\omega_E(u,x)+1)(\omega_E(v,y)+1)}{\omega_E(u,v) \omega_E(x,y)}$ (line \ref{algccm:ratiostart}), 
since the number of equivalent CDES from $G$ to $H$ is $\omega_E(u,v) \omega_E(x,y)$, and the number of CDES from $H = (V , E^\prime , \L , c)$ to $G$ is 
$\omega_{E^\prime}(u,x) \omega_{E^\prime}(v,y) = (\omega_E(u,x)+1)(\omega_E(v,y)+1)$.

The statement follows from the irreducibility and aperiodicity properties of the Markov chain $\G$ explored by \algname, proved by Theorems~\ref{thm:irreducible} and~\ref{thm:aperiodic}; 
those imply that $\G$ is ergodic, and the use of the MH transition steps guarantees convergence to the stationary distribution $\pi$.
\end{proof}

\ifextversion
\begin{figure*}[ht]
\begin{subfigure}{.3\textwidth}
  \centering
  \includegraphics[width=\textwidth]{./figures/label-avg-neigh.pdf}
\end{subfigure} \\
\begin{subfigure}{.3\textwidth}
  \centering
  \includegraphics[width=\textwidth]{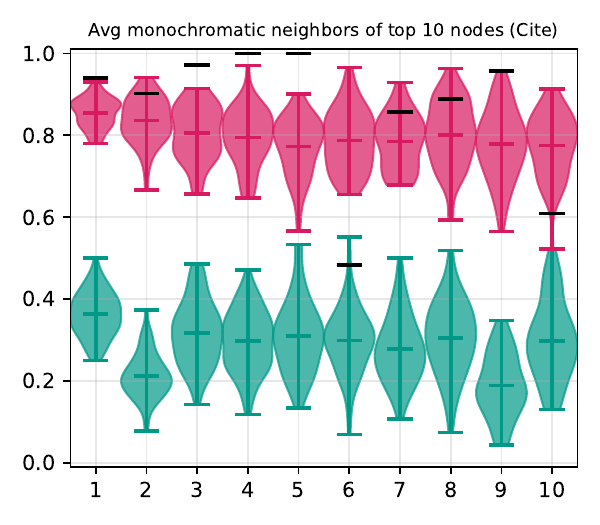} 
\end{subfigure}
\begin{subfigure}{.3\textwidth}
  \centering
  \includegraphics[width=\textwidth]{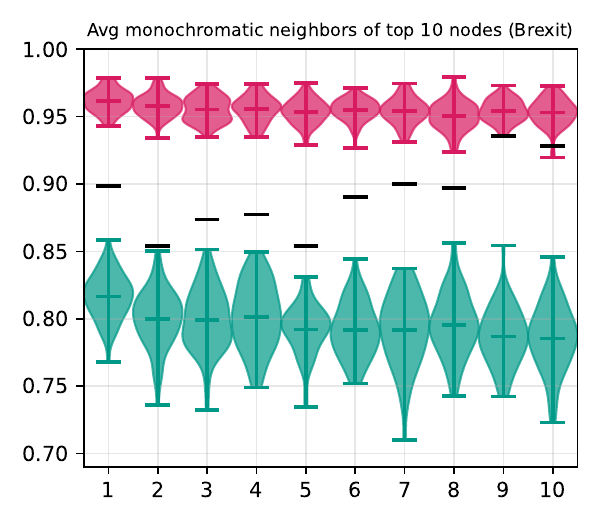} 
\end{subfigure}
\begin{subfigure}{.3\textwidth}
  \centering
  \includegraphics[width=\textwidth]{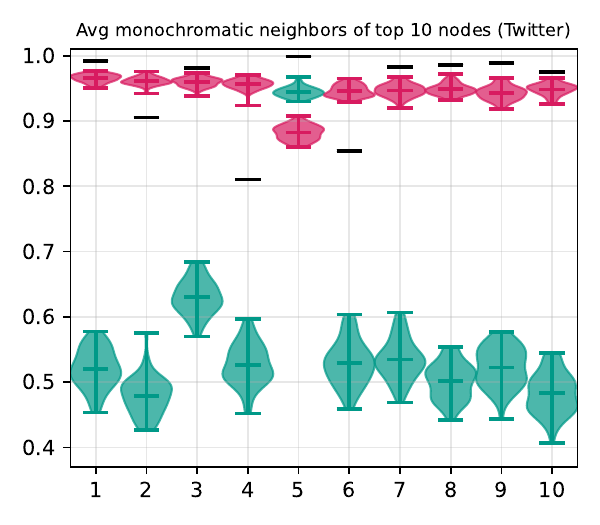} 
\end{subfigure}
\begin{subfigure}{.3\textwidth}
  \centering
  \includegraphics[width=\textwidth]{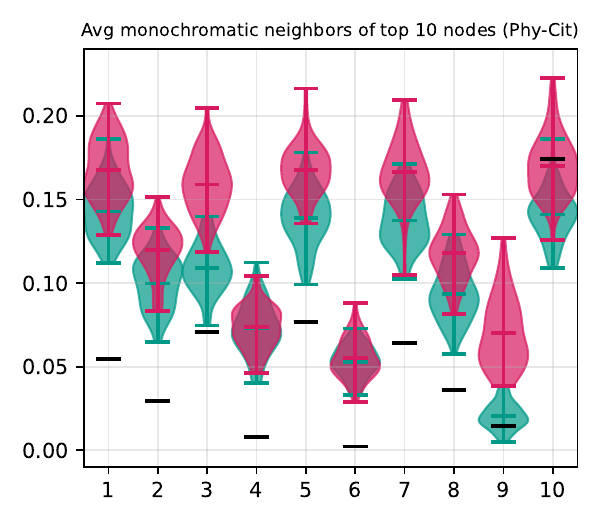} 
\end{subfigure}
\begin{subfigure}{.3\textwidth}
  \centering
  \includegraphics[width=\textwidth]{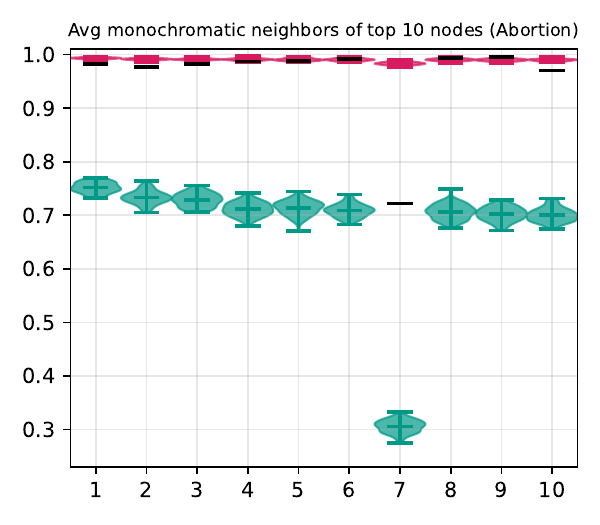} 
\end{subfigure}
\begin{subfigure}{.3\textwidth}
  \centering
  \includegraphics[width=\textwidth]{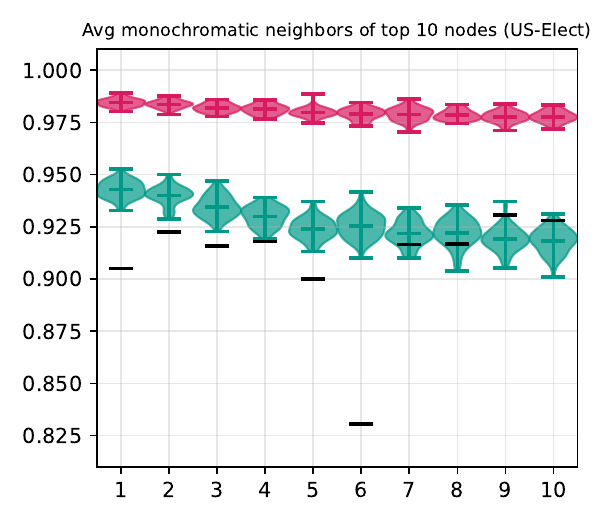} 
\end{subfigure}
\begin{subfigure}{.3\textwidth}
  \centering
  \includegraphics[width=\textwidth]{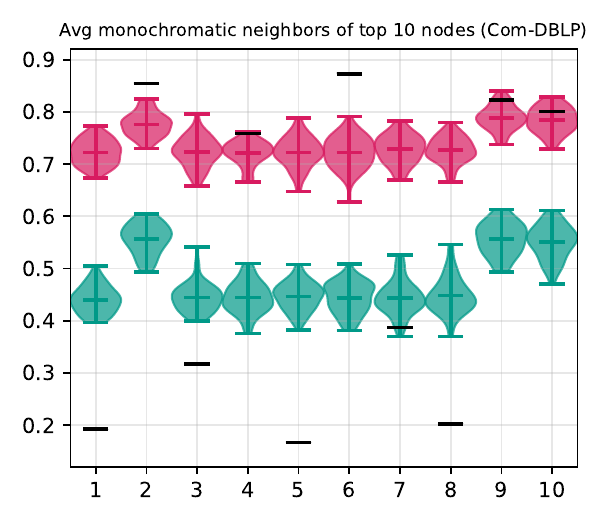} 
\end{subfigure}
\begin{subfigure}{.3\textwidth}
  \centering
  \includegraphics[width=\textwidth]{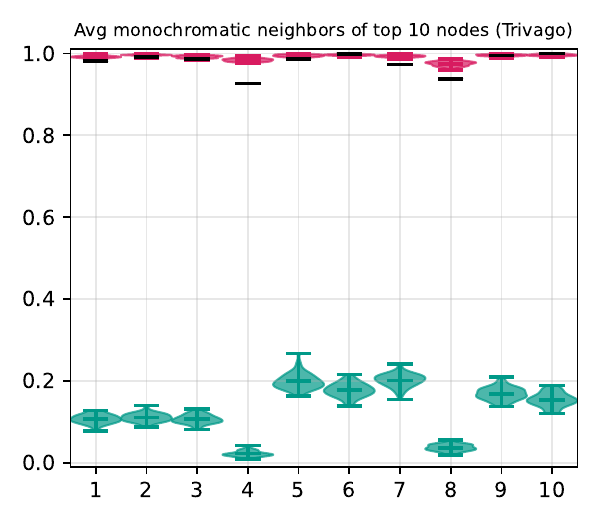} 
\end{subfigure}
\begin{subfigure}{.3\textwidth}
  \centering
  \includegraphics[width=\textwidth]{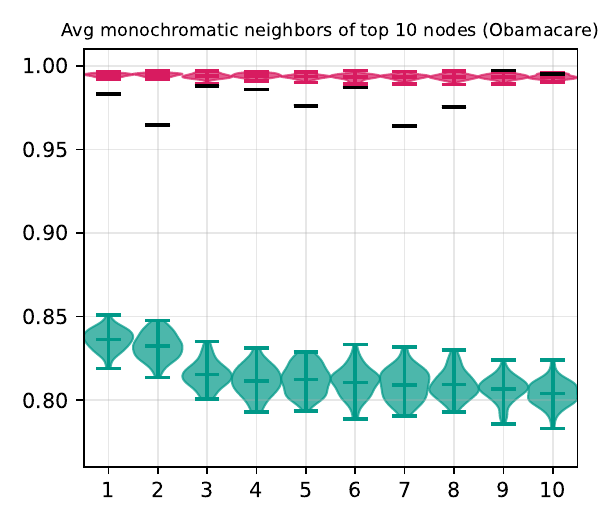} 
\end{subfigure}
\begin{subfigure}{.3\textwidth}
  \centering
  \includegraphics[width=\textwidth]{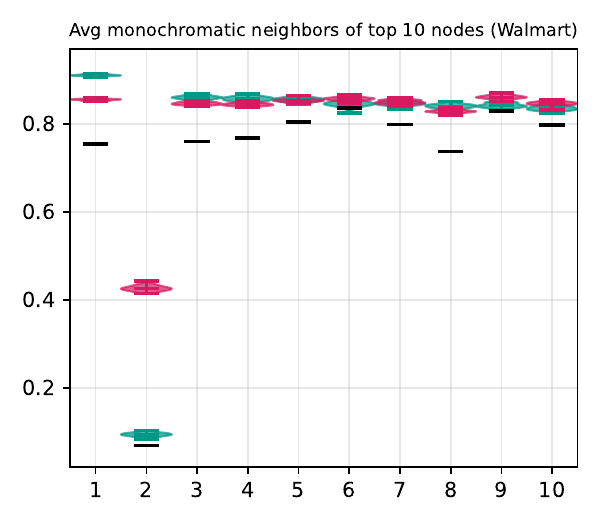} 
\end{subfigure}
\begin{subfigure}{.3\textwidth}
  \centering
  \includegraphics[width=\textwidth]{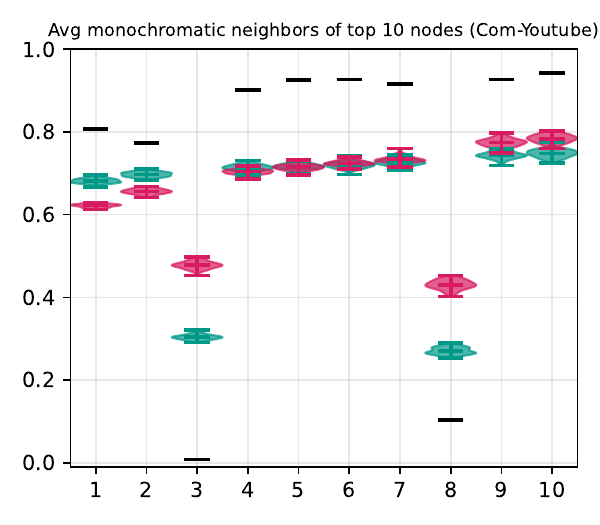} 
\end{subfigure}
\begin{subfigure}{.3\textwidth}
  \centering
  \includegraphics[width=\textwidth]{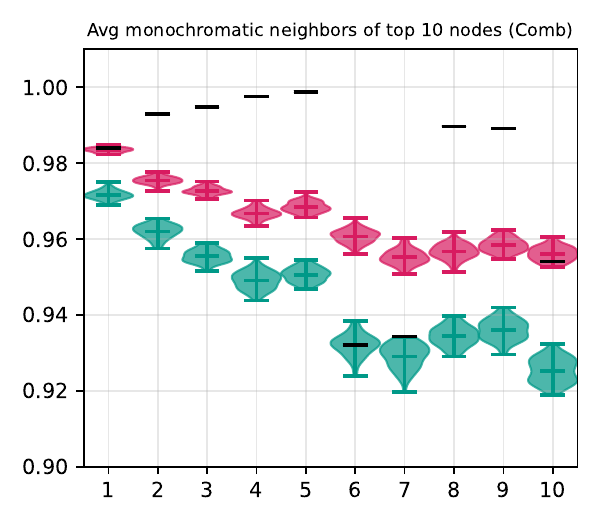} 
\end{subfigure}
\caption{ 
Values of $M_v$ for the $10$ nodes with highest degree, on the observed graphs vs random samples.
}
\Description{This figure shows several plots, reporting the values of $M_v$ for the $10$ nodes with highest degree, on the observed graphs vs random samples.}
\label{fig:avgcolcmappendix}
\end{figure*}

\begin{figure*}[ht]
\begin{subfigure}{.17\textwidth}
  \centering
  \includegraphics[width=\textwidth]{./figures/label-ccm-ccmb.pdf}
\end{subfigure} \\
\begin{subfigure}{.24\textwidth}
  \centering
  \includegraphics[width=\textwidth]{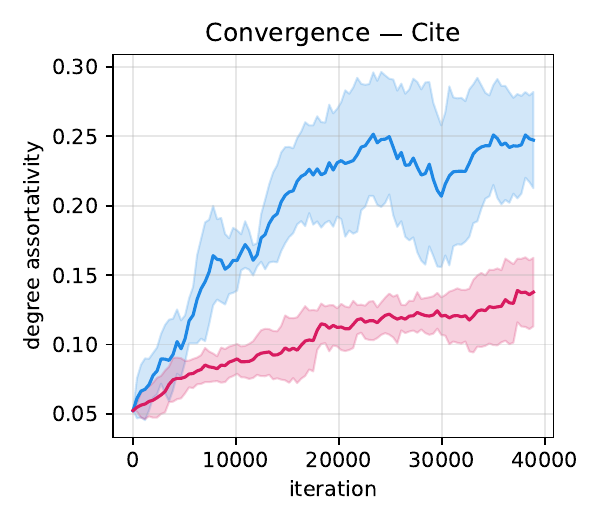} 
\end{subfigure}
\begin{subfigure}{.24\textwidth}
  \centering
  \includegraphics[width=\textwidth]{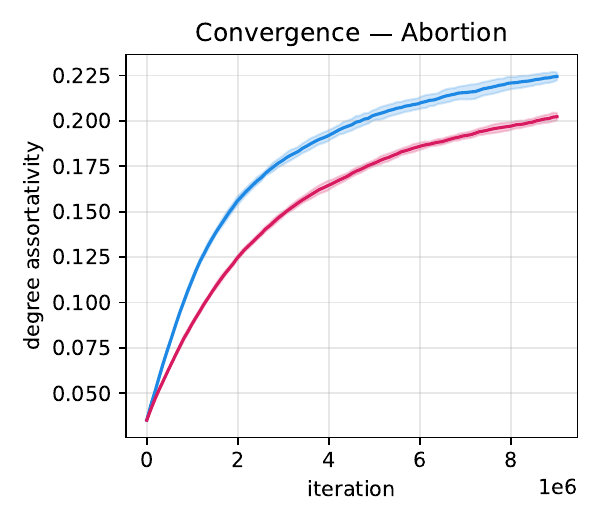} 
\end{subfigure}
\begin{subfigure}{.24\textwidth}
  \centering
  \includegraphics[width=\textwidth]{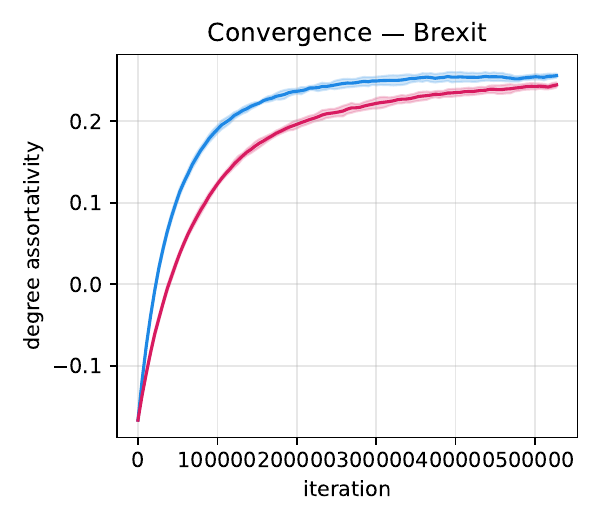}
\end{subfigure}
\begin{subfigure}{.24\textwidth}
  \centering
  \includegraphics[width=\textwidth]{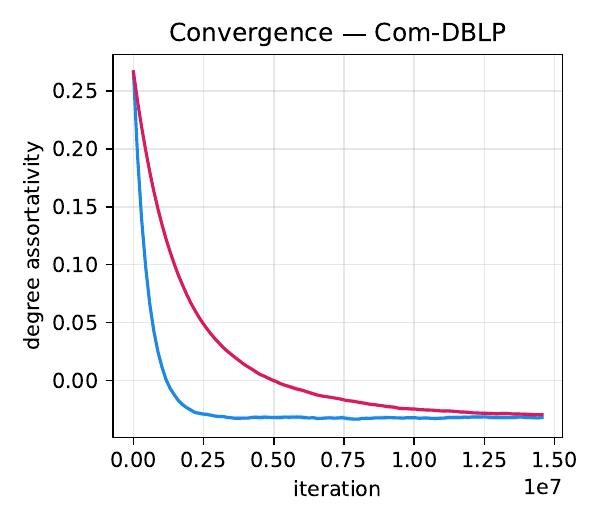}
\end{subfigure}
\begin{subfigure}{.24\textwidth}
  \centering
  \includegraphics[width=\textwidth]{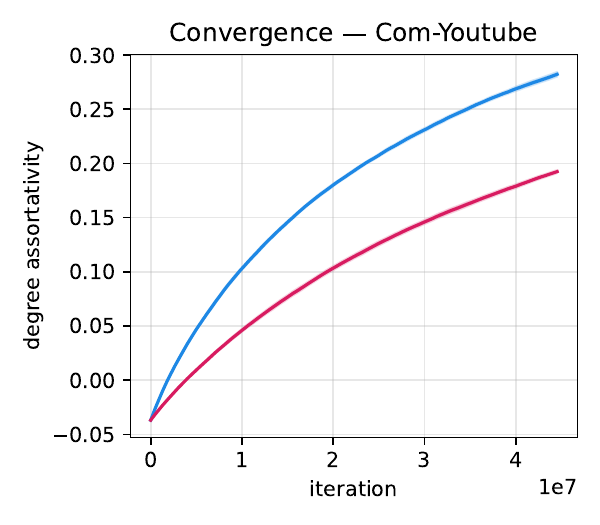}
\end{subfigure}
\begin{subfigure}{.24\textwidth}
  \centering
  \includegraphics[width=\textwidth]{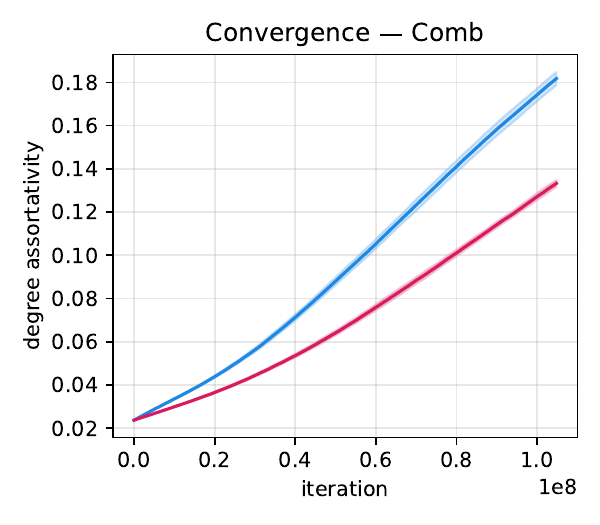}
\end{subfigure}
\begin{subfigure}{.24\textwidth}
  \centering
  \includegraphics[width=\textwidth]{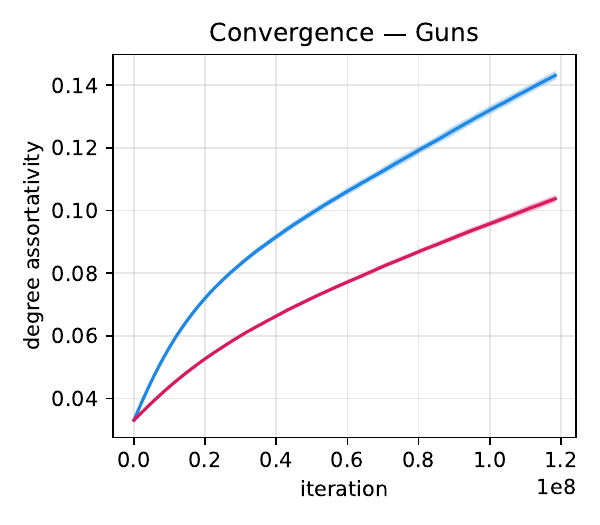}
\end{subfigure}
\begin{subfigure}{.24\textwidth}
  \centering
  \includegraphics[width=\textwidth]{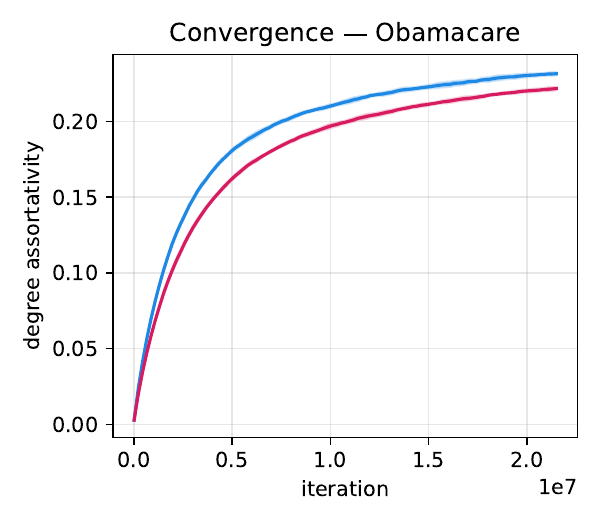}
\end{subfigure}
\begin{subfigure}{.24\textwidth}
  \centering
  \includegraphics[width=\textwidth]{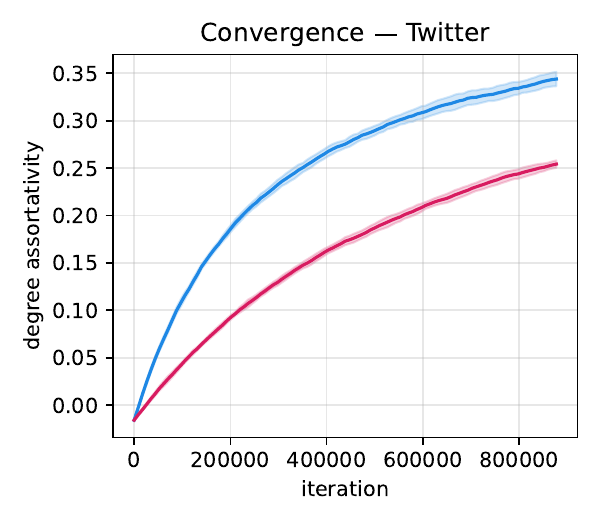}
\end{subfigure}
\begin{subfigure}{.24\textwidth}
  \centering
  \includegraphics[width=\textwidth]{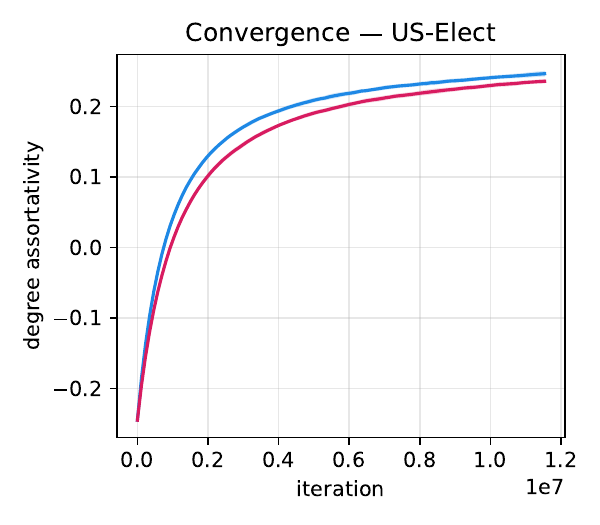}
\end{subfigure}
\begin{subfigure}{.24\textwidth}
  \centering
  \includegraphics[width=\textwidth]{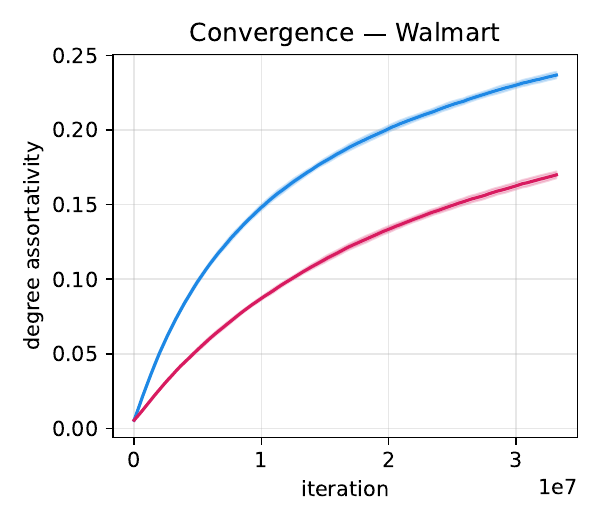}
\end{subfigure}
\caption{  Degree assortativity of \algname\ and \algname-B up to $m \ln m$ iterations.
}
\Description{This figure shows the degree assortativity of \algname\ and \algname-B up to $m \ln m$ iterations.}
\label{fig:convergenceccmappx}
\end{figure*}

\begin{figure*}[ht]
\begin{subfigure}{.17\textwidth}
  \centering
  \includegraphics[width=\textwidth]{./figures/label-ccm-ccmb.pdf}
\end{subfigure} \\
\begin{subfigure}{.24\textwidth}
  \centering
  \includegraphics[width=\textwidth]{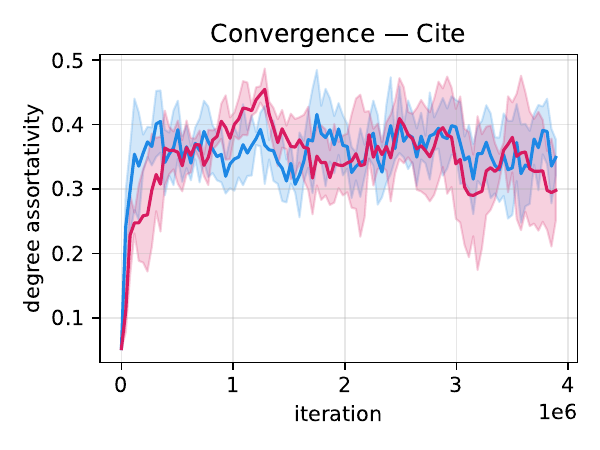} 
\end{subfigure}
\begin{subfigure}{.23\textwidth}
  \includegraphics[width=\textwidth]{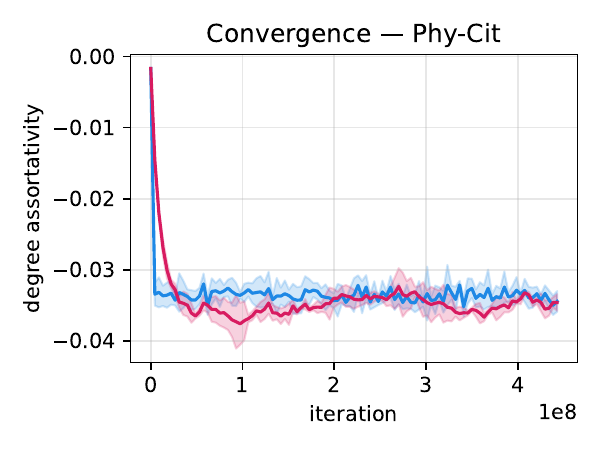}
\end{subfigure}
\begin{subfigure}{.24\textwidth}
  \centering
  \includegraphics[width=\textwidth]{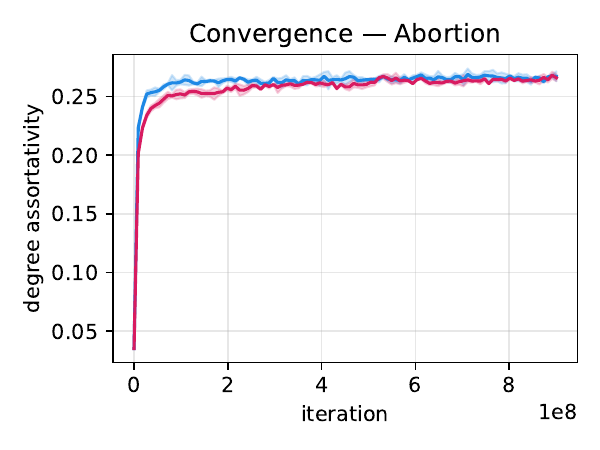} 
\end{subfigure}
\begin{subfigure}{.24\textwidth}
  \centering
  \includegraphics[width=\textwidth]{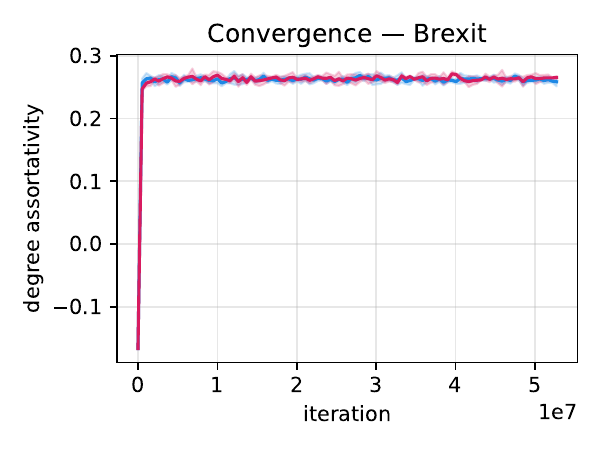}
\end{subfigure}
\begin{subfigure}{.24\textwidth}
  \centering
  \includegraphics[width=\textwidth]{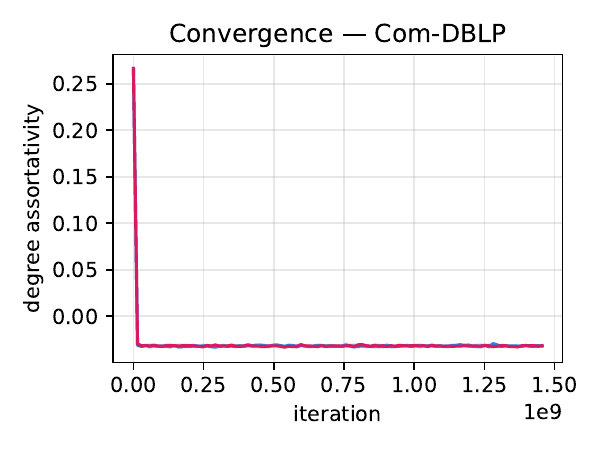}
\end{subfigure}
\begin{subfigure}{.24\textwidth}
  \centering
  \includegraphics[width=\textwidth]{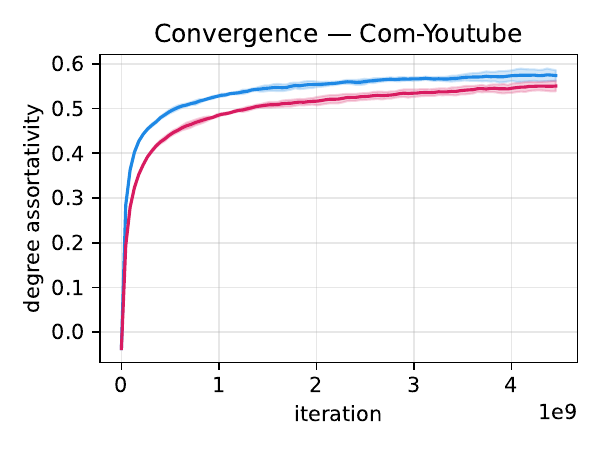}
\end{subfigure}
\begin{subfigure}{.24\textwidth}
  \centering
  \includegraphics[width=\textwidth]{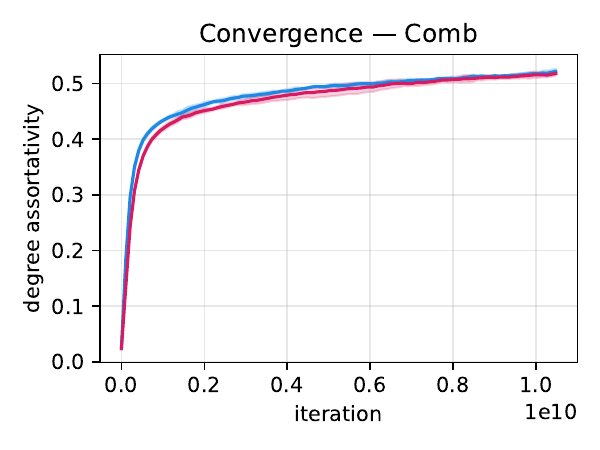}
\end{subfigure}
\begin{subfigure}{.24\textwidth}
  \centering
  \includegraphics[width=\textwidth]{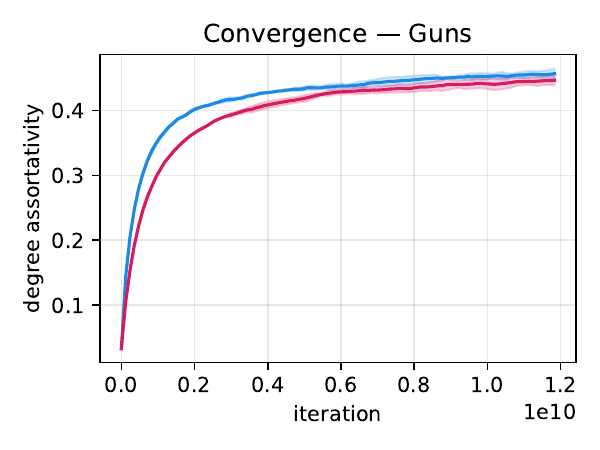}
\end{subfigure}
\begin{subfigure}{.24\textwidth}
  \centering
  \includegraphics[width=\textwidth]{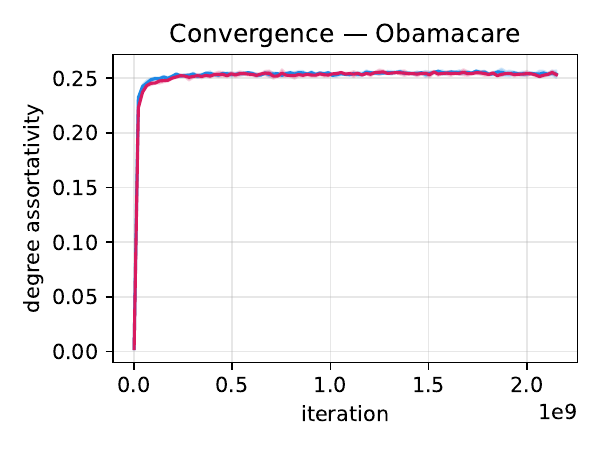}
\end{subfigure}
\begin{subfigure}{.24\textwidth}
  \centering
  \includegraphics[width=\textwidth]{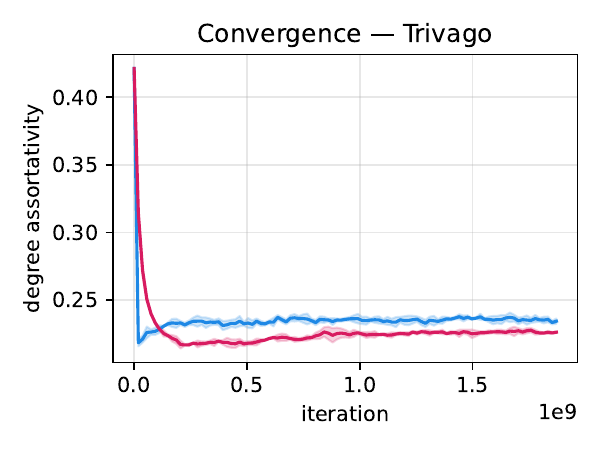}
\end{subfigure}
\begin{subfigure}{.24\textwidth}
  \centering
  \includegraphics[width=\textwidth]{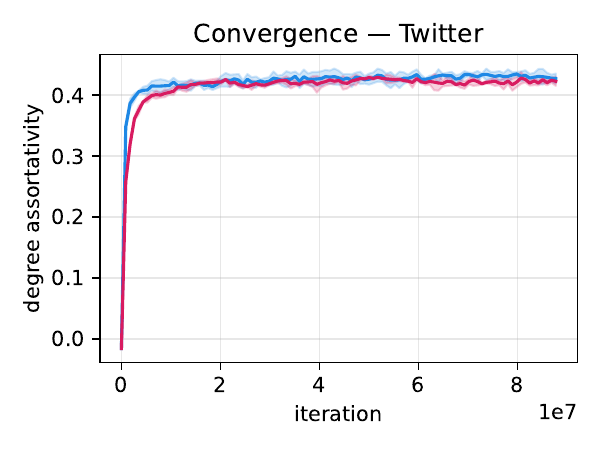}
\end{subfigure}
\begin{subfigure}{.24\textwidth}
  \centering
  \includegraphics[width=\textwidth]{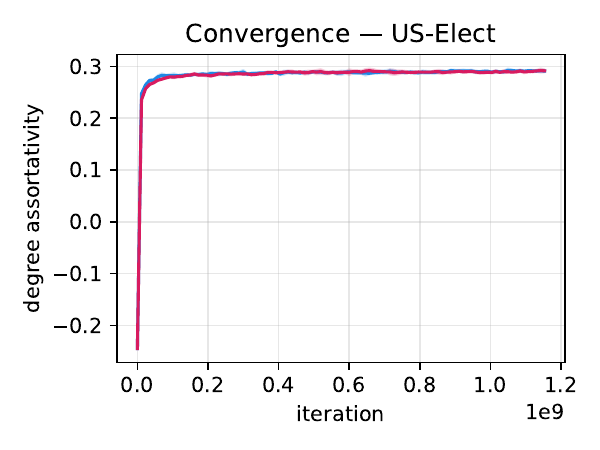}
\end{subfigure}
\begin{subfigure}{.24\textwidth}
  \centering
  \includegraphics[width=\textwidth]{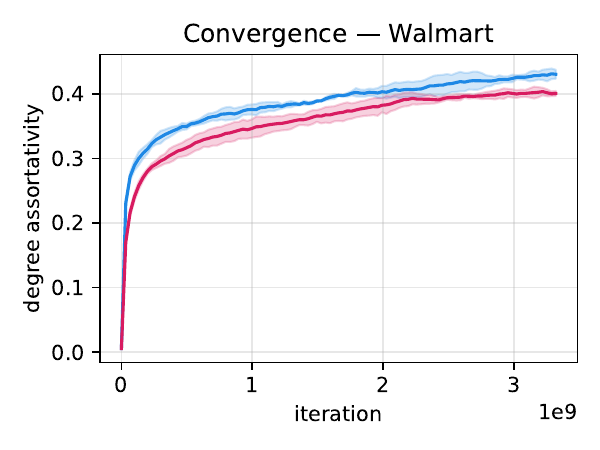}
\end{subfigure}
\caption{  Degree assortativity of \algname\ and \algname-B up to $10^2 m \ln m$ iterations.
}
\Description{This figure shows the degree assortativity of \algname\ and \algname-B up to $10^2 m \ln m$ iterations.}
\label{fig:convergenceccmappxmoreit}
\end{figure*}

\begin{figure*}[ht]
\begin{subfigure}{.65\textwidth}
  \centering
  \includegraphics[width=.7\textwidth]{./figures/labels-iterations-stats.pdf}
\end{subfigure}  \\
\begin{subfigure}{.195\textwidth}
  \includegraphics[width=\textwidth]{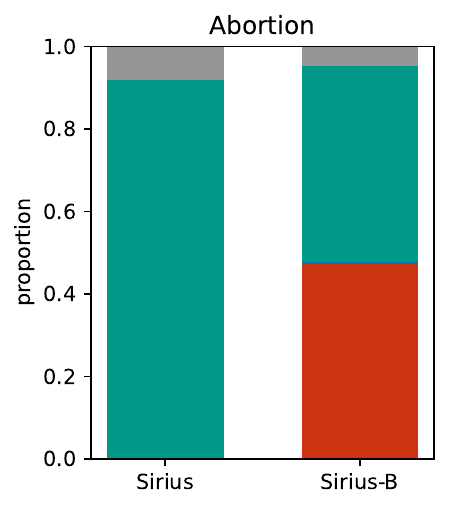}
\end{subfigure}
\begin{subfigure}{.195\textwidth}
  \includegraphics[width=\textwidth]{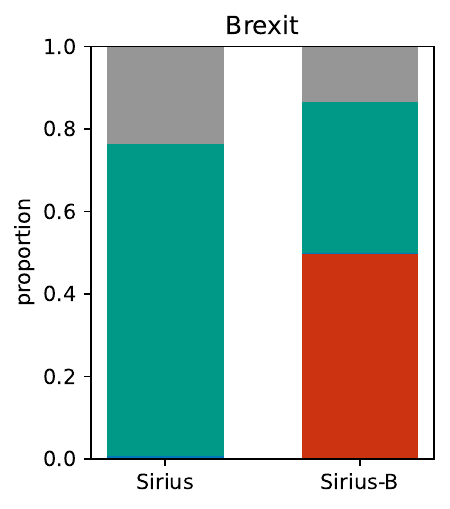}
\end{subfigure}
\begin{subfigure}{.195\textwidth}
  \centering
  \includegraphics[width=\textwidth]{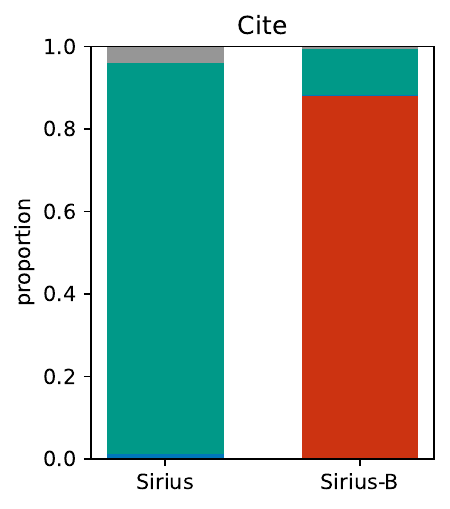} 
\end{subfigure}
\begin{subfigure}{.195\textwidth}
  \centering
  \includegraphics[width=\textwidth]{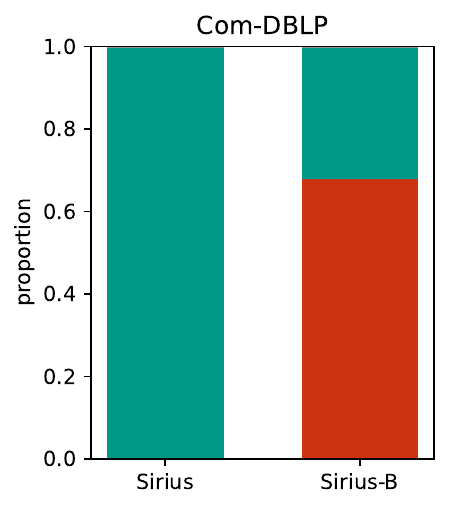}
\end{subfigure}
\begin{subfigure}{.195\textwidth}
  \centering
  \includegraphics[width=\textwidth]{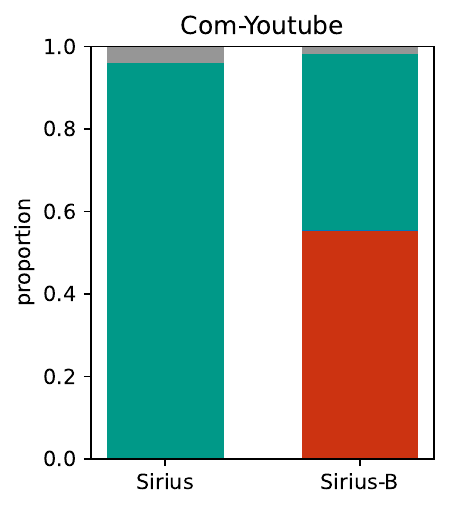}
\end{subfigure}
\begin{subfigure}{.195\textwidth}
  \centering
  \includegraphics[width=\textwidth]{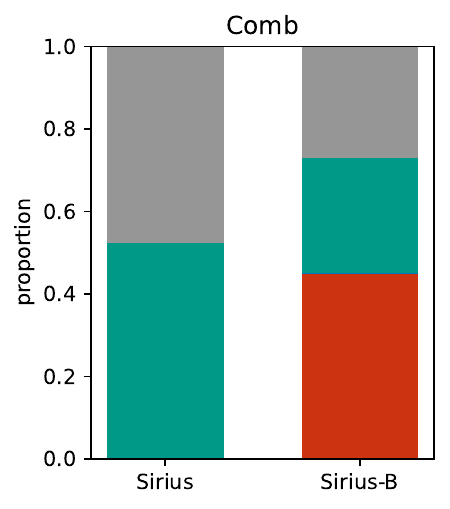}
\end{subfigure}
\begin{subfigure}{.195\textwidth}
  \includegraphics[width=\textwidth]{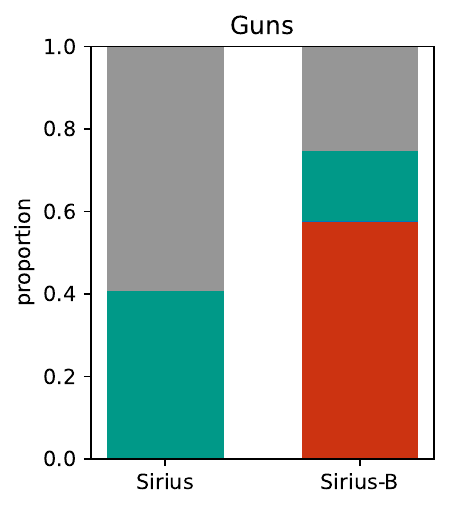}
\end{subfigure}
\begin{subfigure}{.195\textwidth}
  \centering
  \includegraphics[width=\textwidth]{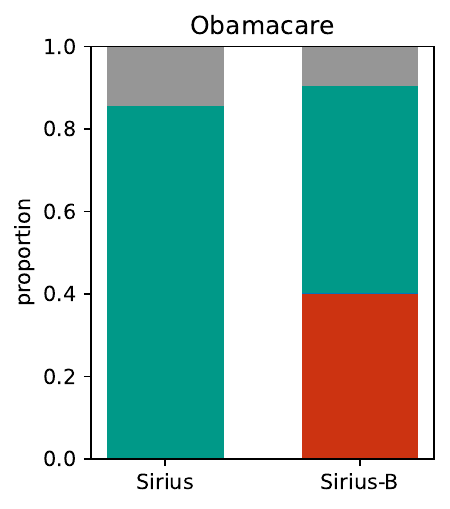} 
\end{subfigure}
\begin{subfigure}{.195\textwidth}
  \centering
  \includegraphics[width=\textwidth]{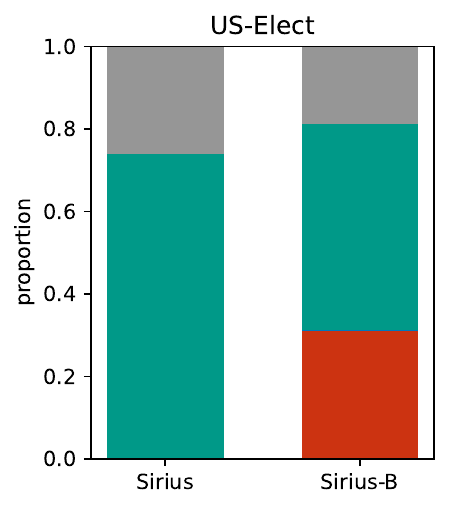}
\end{subfigure}
\begin{subfigure}{.195\textwidth}
  \centering
  \includegraphics[width=\textwidth]{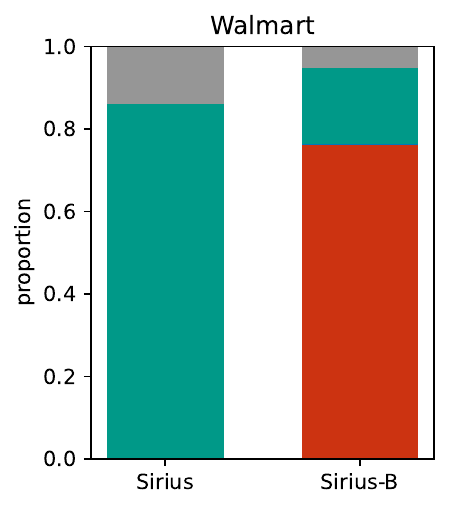}
\end{subfigure}
\caption{  Iteration outcomes for \algname\ and \algname-B. 
}
\Description{This figure shows iteration outcomes for \algname\ and \algname-B.}
\label{fig:iterstatsappx}
\end{figure*}

\begin{figure*}[ht]
\begin{subfigure}{.48\textwidth}
  \centering
  \includegraphics[width=.4\textwidth]{./figures/label-ccm-ccmb.pdf}
\end{subfigure}
\begin{subfigure}{.48\textwidth}
  \centering
  \includegraphics[width=0.6\textwidth]{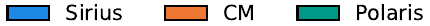}
\end{subfigure} \\
\begin{subfigure}{.415\textwidth}
  \centering
  \includegraphics[width=\textwidth]{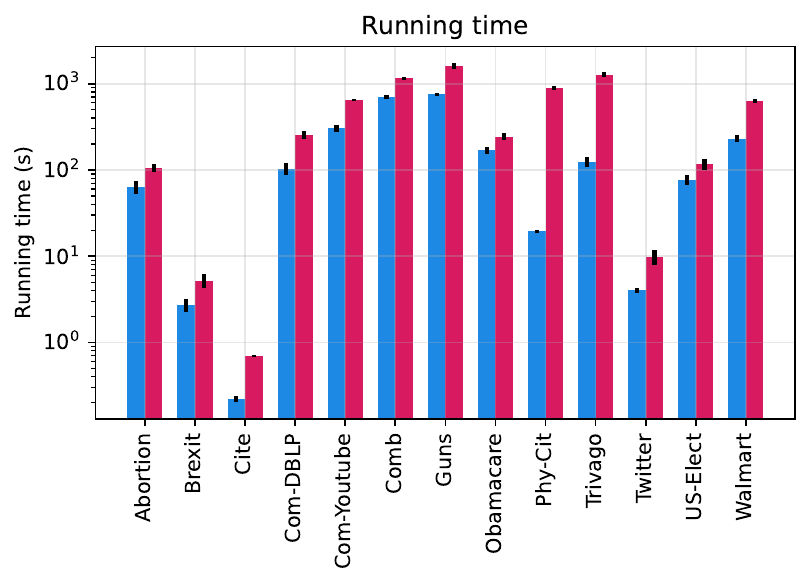} 
    \caption{}
\end{subfigure}
\begin{subfigure}{.45\textwidth}
  \centering
  \includegraphics[width=\textwidth]{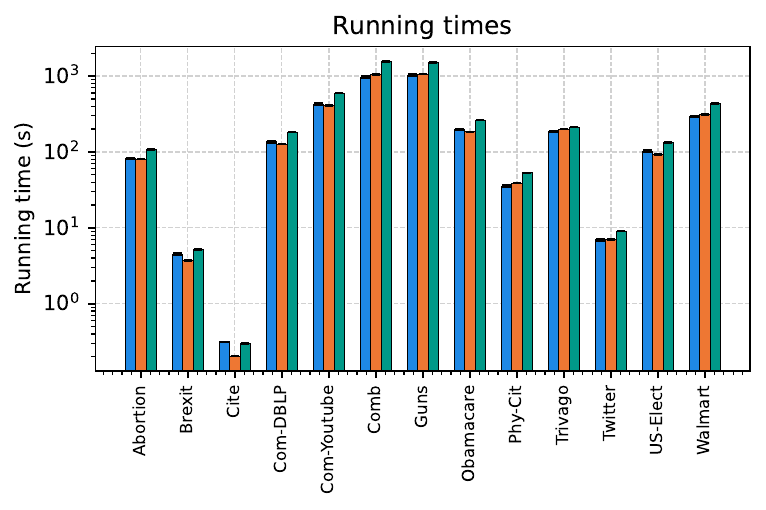} 
    \caption{}
\end{subfigure}
\caption{  Running time comparison between \algname\ and \algname-B (a) 
and with CM and Polaris (b). 
}
\Description{This figure shows the running time comparison between \algname\ and \algname-B (a) 
and with CM and Polaris (b).}
\label{fig:runtimesall}
\end{figure*}

\fi

\end{document}